\documentclass[twocolumn,9pt]{article}
\usepackage[centertags]{amsmath}
\usepackage{amssymb}
\usepackage{url}
\usepackage{graphicx}
\graphicspath{{Figures/}}
\usepackage[table]{xcolor}
\usepackage{verbatim} 
\usepackage[export]{adjustbox}
\usepackage{float}
\usepackage{caption}
\usepackage{subcaption}
\usepackage{mathtools}
\usepackage{amsthm}
\usepackage{cancel}
\usepackage{hyperref}
\usepackage{pdfpages}
\usepackage{hyperref}
\usepackage{graphicx}
\usepackage{threeparttable}
\graphicspath{{Figures/}}
\usepackage[export]{adjustbox}
\usepackage{float}
\usepackage{caption}
\usepackage{subcaption}
\usepackage{mathtools}
\usepackage{amsthm}
\usepackage{cancel}
\usepackage{hyperref}
\usepackage{pdfpages}
\usepackage{hyperref}
\usepackage{bm}

\def\lef[{\left[\begin{array}}
\def\rig]{\end{array}\right]}
\def\qed{\hfill$\Box \Box \Box$}

\def\rea{\mathbb{R}}

\usepackage{psfrag}

\newcommand{\col}{ \mbox{col} }
\newcommand{\rank}{ \mbox{rank } }

\def\rea{\mathbb{R}}

\def\begequ{\begin{equation}}
\def\endequ{\end{equation}}
\def\lab{\label}
\def\begite{\begin{itemize}}
\def\endite{\end{itemize}}
\def\begarr{\begin{array}}
\def\endarr{\end{array}}
\def\begequarr{\begin{eqnarray}}
\def\endequarr{\end{eqnarray}}
\def\diag{\mbox{diag}}
\usepackage{multicol}

\def\caly{{\cal Y}}

\def\L2{{\cal L}_2}
\def\L2e{{\cal L}_{2e}}

\def\rea{\mathbb{R}}

\def\adj{\mbox{adj}}

\def\diag{\mbox{diag}}

\def\diag{\mbox{diag}}

\def\adj{\mbox{adj}}
\def\col{\mbox{col}}

\def\diag{\mbox{diag}}
\def\rank{\mbox{rank}\;}

\def\begmat#1{\begin{bmatrix}#1\end{bmatrix}}

\def\begsubequ{\begin{subequations}}
\def\endsubequ{\end{subequations}}
\def\begequarr{\begin{eqnarray}}
\def\endequarr{\end{eqnarray}}
\def\begequarrs{\begin{eqnarray*}}
\def\endequarrs{\end{eqnarray*}}
\def\begarr{\begin{array}}
\def\endarr{\end{array}}
\def\begequ{\begin{equation}}
\def\endequ{\end{equation}}
\def\lab{\label}
\def\begdes{\begin{description}}
\def\enddes{\end{description}}
\def\begenu{\begin{enumerate}}
\def\begite{\begin{itemize}}
\def\endite{\end{itemize}}
\def\endenu{\end{enumerate}}

\def\lef[{\left[\begin{array}}
\def\rig]{\end{array}\right]}
\def\qed{\hfill$\Box \Box \Box$}
\def\begcen{\begin{center}}
\def\endcen{\end{center}}
\def\begrem{\begin{remark}\rm}
\def\endrem{\end{remark}}
\def\begassums{\begin{assumption*}}
\def\endassums{\end{assumption*}}
\def\begassu{\begin{ass}}
\def\endassu{\end{ass}}
\def\beglem{\begin{lemma}}
\def\endlem{\end{lemma}}
\def\begcor{\begin{corollary}}
\def\endcor{\end{corollary}}
\def\begfac{\begin{fact}}
\def\endfac{\end{fact}}
\def\begass{\begin{assumption}}
\def\endass{\end{assumption}}

\def\begmat#1{\begin{bmatrix}#1\end{bmatrix}}

\def\begenu{\begin{enumerate}}
	\def\begite{\begin{itemize}}
		\def\endite{\end{itemize}}
	\def\endenu{\end{enumerate}}

\usepackage[noadjust]{cite}
\usepackage{lipsum}
\usepackage{color}
\usepackage{mathtools, cuted}
\usepackage{lipsum, color}

\graphicspath{{Images}}
\hypersetup{
    colorlinks=true,
    linkcolor=blue,
    filecolor=blue,
    urlcolor=blue,
    citecolor=blue,
    }

\newtheorem{remark}{Remark}
\newtheorem{assumption}{Assumption}
\newtheorem{property}{Property}
\newtheorem{defin}{Definition}
\newtheorem{proposition}{Proposition}
\newtheorem{lemma}{Lemma}
\newtheorem{thm}{Theorem}

\title{Three High Performance Global Tracking Composite Adaptive Controllers for Fully Actuated Euler-Lagrange Systems: Experimental Validation} 

\author{Luis Cervantes-Pérez\textsuperscript{1}, Jose Guadalupe Romero*\textsuperscript{2},\\
 Romeo Ortega\textsuperscript{2}, Víctor Santibáñez\textsuperscript{1} and Jesús Sandoval\textsuperscript{3},\\
  \small 1 Tecnológico Nacional de México / I. T. de La Laguna,\\
  \small Torreón, Coahuila, México \\
 \small 2 Departamento de Ingeniería Eléctrica y Electrónica / \\
 \small Instituto Tecnológico Autónomo de México, CDMX, Ciudad de México, México \\
  \small 3 Tecnológico Nacional de México / I. T. de La Paz,\\
 \small La Paz, Baja California Sur, México \\
 }
\usepackage{titlesec}
\date{}

\titleformat{\paragraph}
{\normalfont\normalsize\bfseries}{\theparagraph}{1em}{}
\titlespacing*{\paragraph}
{0pt}{3.25ex plus 1ex minus .2ex}{1.5ex plus .2ex}

\begin{document}

\twocolumn[
  \begin{@twocolumnfalse}
    \maketitle
    \begin{abstract}
Three adaptive global tracking controllers for fully actuated Euler–Lagrange systems, with verifiable {\em performance improvement} over existing designs, are reported in this letter. Two of these controllers ensure global {\em exponential} convergence under a {\em weak interval excitation} condition. Besides, one of the proposed controllers features a simple adaptive PID-like structure that—unlike classical solutions—avoids the need for additional filtering. We adopt a composite adaptation architecture, invoke a novel parameterization of the system dynamics and use a high performance estimation scheme recently introduced in the literature. Real-time experiments and a comparative study with a learning-based adaptive controller on a two–degrees-of-freedom manipulator arm illustrate the effectiveness of the proposed controllers.
    \end{abstract}
  \end{@twocolumnfalse}
]

\section{Introduction}\lab{sec1}
The problem of designing adaptive global tracking controllers for fully actuated Euler-Lagrange (EL) systems has been investigated for a long time, with numerous solutions being reported, for example in:  \cite{ORTetalbook,SLOLItac,SLOLIaut,SPOHUTVIDbook}. To prove the good performance of these controllers---{\em e.g.}, parameter convergence or exponential stability---it is necessary to assume strong excitation conditions, for instance, {\em persistency of excitation} (PE). It is well-known that PE is a very strong requirement that, for instance, is {\em not satisfied} in regulation tasks. An experimental evaluation of several parameter estimation methods for a robot manipulator was presented in \cite{BERGHUIS}, where the advantages and disadvantages of the {\em classical methods} are highlighted, with excitation being the major obstacle.
After \cite{SLOLItac}, a large number of adaptive trajectory tracking controllers have been published up to this day. For example \cite{CHELIUSLO,YANGetal,BRAGetal,DIXetal}. Notable advances in addressing this limitation have been achieved using {\em learning-based}  and neural networks  techniques  \cite{ZHOUeta, DUCetal, JIN, GUO2020}; however, the computational cost increases significantly, rendering high-computation-demand adaptive controllers difficult to implement.
\par A survey on composite adaptation and learning techniques to improve parameter convergence in adaptive control for robot control applications is presented in  \cite{GUOPAN}. Unfortunately, the majority of works mentioned previously have a complicated parameterization of the linear regression equation (LRE), and parameter convergence is ensured by imposing the unfeasible persistent excitation (PE) condition.  In  this paper we propose some new controllers that overcome this limitation.
\par Although high-quality tracking is indeed the primary objective in robotics, estimating model parameters remains valuable because the parameters encode the system’s physical characteristics. As discussed in \cite{raol}, parameter estimation from input--output data is an important step in the analysis of dynamic systems, since the parameters often describe the stability and control behaviour of the system and help establish the adequacy of the mathematical model. Knowing the system parameters can also be important to compute manipulability indices \cite{yoshikawa1985,yoshikawa1990} and posture indices \cite{pamanes2018}, since these performance measures typically require full knowledge of the system parameters.

\par We present in the paper three new controllers with verifiable {\em performance improvement} assuming only the {\em weak interval excitation} (IE) condition  \cite{KRERIE} that, as is well-known is strictly weaker than PE and satisfied for all standard control tasks---including regulation. Under the IE assumption, we establish global {\em exponential} convergence of the position and velocity tracking errors, as well as of the estimated uncertain parameters, for two of these controllers; whereas for the third one, we prove that the velocity and parameter estimation errors converge asymptotically to zero, while the tracking error remains bounded.

The following key modifications are instrumental for the success of our designs.

\begin{enumerate}
\item[$\mathbf{M1}$] The use of a {\em composite adaptation architecture} \cite{DUANAR,PANORTMOY,SLOLIaut}---that is known to ensure robustness properties which are stronger than the one of the classical direct or indirect schemes
\item[$\mathbf{M2}$] We invoke a novel parameterization of the system dynamics \cite{ROMORTBOB}, which has shown to be more effective than the classical one \cite{KHOKAN}.
\item[$\mathbf{M3}$] We use a the high performance least squares plus dynamic regression extension and mixing (LS+DR\-EM) estimation scheme recently reported in \cite{ORTROMARA}, which ensures exponential convergence of the estimates, imposing only the extremely weak IE assumption. Moreover, it has shown to be robust in the presence of noise---an essential feature in the application at hand.
\end{enumerate}
In comparison with existing literature the main contributions are:
\begin{enumerate}
\item[$\mathbf{C1}$] Compared with our previous result \cite{ROMORT}, this note presents an additional adaptive PID-like controller based on \emph{measured} tracking errors, which ensures exponential convergence. Regarding one of the previous solutions reported new stability conclusions are derivated using Matrosov's theorem.
\item[$\mathbf{C2}$] We propose three composite adaptive trajectory-track\-ing controllers that achieve the tracking control objective and guarantee convergence of the parameter estimates to their true values without imposing the persistent excitation condition.
\item[$\mathbf{C3}$] Compared with composite adaptive controllers that do not impose persistent excitation, and in particular with learning-based methods such as \cite{guo2018composite2}, the proposed schemes have verifiably lower computational cost.
\item[$\mathbf{C4}$] Compared with the classical PD$+$adaptive controller of \cite{SLOLItac}, two of the proposed controllers incorporate a novel parameterization that enhances robustness in the presence of measurement noise.
\end{enumerate}

The remainder of the paper is organized as follows. Section \ref{sec2} presents the system dynamics and the problem  formulation. In Section \ref{sec3} we give some preliminary material used in the sequel.  The main results of the paper, namely three new adaptive global tracking controllers, are presented in Section \ref{sec4}. The experimental implementation is detailed in Section \ref{secc:ExperimentalEvaluation}. Finally, the conclusions are summarized in Section \ref{secc:Conclusions}.
\\
{\bf Notation.} Throughout the document, real numbers are denoted by lowercase Latin or occasionally Greek italic letters. Vectors are denoted by \textbf{bold} lowercase Latin or Greek letters, and matrices are denoted by uppercase Latin letters or occasionally uppercase Greek letters. The notation $\lambda_{\min} \{A\}$ and $\lambda_{\max} \{A\}$ will be used to denote, respectively, the smallest and largest eigenvalue of a positive definite symmetric bounded matrix $A(\bm{x})$ for each $\bm{x}\in\mathbb{R}^n$ (denoted by $A>0$). The notation $A\geq0$ means that the matrix $A$ is positive semi-definite. The Euclidean norm for a vector $\bm{x}$ is defined as $\|\bm{x}\|=(\bm{x}^T\bm{x})^{1/2}$, and for a matrix $A$ as the induced norm: $\|A\| = (\lambda_{\max}\{A^T A\})^{1/2}$. In addition, $(\cdot)_{n\times n}$ denotes a matrix of dimensions $n\times n$, with $I_{n\times n}$ as the identity matrix and $0_{n\times n} $ as an array of zeros. ${\mathbb{R}}_+$ and ${\mathbb{Z}}_+$ denote the positive real and integer numbers, respectively. {\color{black}The expression} $\bm{0}_{n}\in\mathbb{R}^n$ {\color{black}denotes} a vector of zeros of dimension $n\times 1$, and $\mathrm{det}[A]$ {\color{black}represents} the determinant of a square matrix $A$. The absolute value of a real number $x$ is denoted by the standard symbol $|x|$. The gradient transpose operator is defined as $\nabla_{\bm{q}} (\cdot):=({\partial (\cdot) \over \partial \bm{q}})^\top$. The notation $\mathrm{diag}\{a_1, a_2, \ldots, a_n\}$ denotes a diagonal matrix of size $n\times n$ with diagonal elements $a_1, a_2, \ldots, a_n$. Finally, $\mathcal{L}_2$ and $\mathcal{L}_\infty$ denote the spaces of square-integrable functions and bounded signals, respectively.
\section{System Dynamics and Control Problem Formulation}
\lab{sec2}
\subsection{System dynamics}
\lab{subsec21}
In this paper we consider an $n_q$-DoF fully actuated EL systems with generalized coordinates $\bm{q}\in \rea^{n_q}$ and control vector $\boldsymbol\tau \in \rea^{n_q}$, whose dynamics is described by the EL equations of motion
\begequ
\lab{elsys}
{ \frac{d}{dt} \left[\nabla_{\dot{\bm{q}}}{ L}(\bm{q},\dot{\bm{q}})\right] - \nabla_{\bm{q}} { L}(\bm{q},\dot{\bm{q}})= \boldsymbol\tau-\nabla{F_R}(\dot{\bm{q}}),}
\endequ
where ${ L}:\rea^{n_q} \times \rea^{n_q} \to \rea$ is the Lagrangian function
$$
{ L}(\bm{q},\dot{\bm{q}}) :=\frac{1}{2}\dot{\bm{q}}^{\top}M(\bm{q})\dot{\bm{q}}  - {U}(\bm{q}),
$$
with $M:\rea^{n_q} \to \rea^{n_q \times n_q}$, $M(\bm{q}) > 0,$ the generalized inertia matrix and  ${U}:\rea^{n_q} \to \rea$ the potential energy function, and ${F_R}(\dot{\bm{q}}):\mathbb{R}^{n_q}\to \mathbb{R}$ is a dissipative energy function, which is assumed to be linear with respect to a set of constant dynamic parameters $\boldsymbol\theta$. See \cite{ORTetalbook} for additional details on this model and many practical examples, and \cite{SPOHUTVIDbook} for a detailed description of robot manipulators. For a robot manipulator with only viscous friction at its joints, the function $F_R$ can be modeled as the Rayleigh dissipation function given by
\begin{equation}\label{eqn:F}
{F}_R=\frac{1}{2}\dot{\bm{q}}^T\mathcal{R}\dot{\bm{q}}=\frac{1}{2}\sum_{i=1}^{n_q}\dot{q}_i^2f_{v_i},
\end{equation}
where $\mathcal{R}=\mathrm{diag}\{f_{v_1},\ldots,f_{v_{n_q}}\}$ is a diagonal positive definite matrix, with $f_{v_i}>0$ denoting the viscous friction parameter of the $i$th joint, for $i\in\{1,\ldots,n_q\}$.In this work, we consider systems for which $F_R$ is defined as in \eqref{eqn:F}. Under these considerations, the dynamics of the EL system \eqref{elsys} can be written as
$$\lab{dynmodel}
{d \over dt }\left[  M(\bm{q}) \dot{\bm{q}} \right] - \frac{1}{2} \nabla_{\bm{q}} \left[  \dot{\bm{q}}^{\top}M(\bm{q})\dot{\bm{q}}\right]   + \nabla {U}(\bm{q})+\mathcal{R}\dot{\bm{q}} = \boldsymbol\tau,
$$
with more explicit and well-known form
\begin{equation}
\label{robdyn}
M(\bm{q}) \ddot{\bm{q}}+ C(\bm{q}, \dot{\bm{q}})\dot{\bm{q}} + \nabla{{U}}(\bm{q})+\mathcal{R}\dot{\bm{q}} = \boldsymbol\tau,
\end{equation}
where $C: \rea^{n_q} \times \rea^{n_q} \to \rea^{n_q \times n_q} $ represents the Coriolis and centrifugal forces matrix. In this study, it is assumed that $\bm{q}$ and $\dot{\bm{q}}$ are measurable. The following two assumptions related to the inertia matrix and the potential energy function are considered:
\begin{assumption}\label{assum:M}
The system inertia matrix $M(\bm{q})$ is symmetric, positive definite, twice continuously differentiable, and its maximum and minimum singular values are bounded and bounded away from zero, respectively.
\end{assumption}
\begin{assumption}\label{assum:u}
The system potential energy function ${U}(\bm{q})$ is continuously differentiable.
\end{assumption}
Moreover, the following properties of the system \eqref{robdyn} with only revolute joints— which satisfies Assumptions \ref{assum:M}–\ref{assum:u}—are assumed to be available \cite{ControlOfRobots}:
\begin{property}\label{pp:1}
If the matrix $C(\bm{q},\dot{\bm{q}})$ is defined via the Christoffel symbols of the first kind, the key skew-symmetry property
\begequ
\lab{skesym}
\bm{z}^\top[\dot M(\bm{q})-2C(\bm{q},\dot{\bm{q}})]\bm{z}=0,\;\forall \bm{z} \in \rea^{n_q},
\endequ
holds globally. And it is also true that
\begin{equation}
\dot{M}(\bm{q})=C(\bm{q},\dot{\bm{q}})+C(\bm{q},\dot{\bm{q}})^T.
\end{equation}
\end{property}
\begin{property}\label{pp:2}
A positive constant $k_C$ exist such that
\begin{equation}
\|C(\bm{q},\dot{\bm{q}})\|\leq k_C \|\dot{\bm{q}}\|, \hspace{0.5cm} \forall \bm{q},\dot{\bm{q}}\in\mathbb{R}^{n_q}.
\end{equation}
\end{property}
\begin{property}\label{pp:4}
The vector $\bm{g}(\bm{q})$ is continuous and bounded for every bounded $\bm{q}$.
\end{property}
\begin{defin}
From \eqref{robdyn}, the following vectorial function is defined:
\begin{equation}\label{fun:f}
\begin{split}
\bm{f}(\bm{q},\dot{\bm{q}},\bm{v},\dot{\bm{v}}):=&M(\bm{q},\boldsymbol\theta)\dot{\bm{v}}+C(\bm{q},\dot{\bm{q}},\boldsymbol\theta)\bm{v}+\bm{g}(\bm{q},\boldsymbol\theta)\\
&+\mathcal{R}(\boldsymbol\theta)\dot{\bm{q}},
\end{split}
\end{equation}
where $\boldsymbol\theta\in\mathbb{R}^{w}$ is the vector of unknown constant parameters.
\end{defin}
\begin{property}\label{prop1}
\cite{ControlOfRobots} The function \eqref{fun:f} can be linearly parameterized as:
\begin{equation}\label{eqn:phi2}
\bm{f}(\bm{q},\dot{\bm{q}},\bm{v},\dot{\bm{v}})=Y(\bm{q},\dot{\bm{q}},\bm{v},\bm{\dot{v}})\boldsymbol\theta,
\end{equation}
where $Y:\mathbb{R}^{n_q}\times \mathbb{R}^{n_q} \times \mathbb{R}^{n_q} \times \mathbb{R}^{n_q} \to \mathbb{R}^{n_q\times w}$.
\end{property}

\subsection{Problem formulation}
\lab{subsec22}
The control problem is to design global tracking adaptive controllers for the EL system \eqref{robdyn} whose transient behavior and robustness  {\em outperforms} the existing solutions to the problem, in particular the classical adaptive controller \cite{SLOLItac} and its robust version \cite{BERORTNIJ}.

To achieve this goal we appeal to the following recent results.
\\
{\bf R1.} We derive the linear regression equation (LRE) instrumental to estimate the system parameters with the new parametrization of EL systems obtained using the {\em power-balance} equation \cite {ROMORTBOB}. It has been shown  that---due to its lower dimension and considerably simpler analytic expressions---this parameterization is more effective than the standard parameterization first reported by \cite{KHOKAN} and used in all adaptive controllers \cite{BERORTNIJ,SLOLItac}.
\\

{\bf R2.} Use the first stage of the LS+DREM procedure recently reported in \cite{ORTROMARA}, to generate, from the previous power-balance LRE, a new set of {\em scalar} LRE. As shown in \cite{ORTROMARA}, the new scalar regressor has the distinguishing feature of satisfying the key {\em persistent excitation} (PE) condition imposing only the {\em weakest} assumption of identifiability of the power-balance LRE \cite{GOOSINbook}. As is well-known, PE of the regressor ensures global {\em exponential} convergence of the parameter estimates that, in its turn, guarantees strong robustness properties for the estimator \cite{SASBODbook}.
\\

Equipped with these two elements we propose three new {\em composite} adaptive controllers, which achieve the desired global tracking objective. Composite adaptive controllers, that combine direct and indirect components, have been reported for linear systems \cite{DUANAR}, for robotic applications \cite{yongpinpan,SLOLIaut} and for a class of nonlinear systems \cite{PANORTMOY}. In all these cases,  been effectively exploited to {\em enhance the robustness} of the adaptive controllers.

%
\section{Background Material}
\lab{sec3}
%
In this section we first present the LRE for general EL systems obtained using the power balance equation \cite {ROMORTBOB}. Then, we use the LS+DREM procedure \cite{ORTROMARA}, to generate a new {\em scalar} LRE---which turns out to have a PE regressor.
\subsection{Derivation of the power-balance linear regression equation}
\lab{subsec31}
Following \cite{ROMORTBOB,Cervantes2026} to generate the new LRE we introduce the following parameterization of the inertia matrix $M(\bm{q})$, the potential energy function ${U}(\bm{q})$, and the dissipation energy function ${F_R}(\dot{\bm{q}})$:
\begin{equation}
\label{parmu}
\begin{split}
    M(\bm{q}) = \sum_{i=1}^{\ell} M_i(\bm{q})\theta_{M_i}, &\hspace{0.5cm}
    {U}(\bm{q}) = \sum_{j=1}^{r} {U}_j(\bm{q})\theta_{U_j},\\
    \dot{\bm{q}}^T\mathcal{R}\dot{\bm{q}}&=\sum_{k=1}^{n_q} \dot{q}_k\theta_{{F}_{R_k}},\\
    \end{split}
\end{equation}
 with matrices $M_i: \mathbb{R}^{n} \to \mathbb{R}^{n \times n}$, functions ${U}_j: \mathbb{R}^{n} \to \mathbb{R}$, where $w:=\ell+r+n_q$ is the total number of {\em unknown} constant parameters of the system $\theta_{M_i}, \theta_{U_j}$ and $\theta_{\mathcal{F}_{R_k}}$, that we group together in a single vector as
\begin{equation}\label{thesthe}
\begin{split}
\boldsymbol\theta:=&
\col(\theta_{M_1}, \cdots, \theta_{M_\ell}, \theta_{U_1},  \cdots, \theta_{U_r},\theta_{F_{R_1}},\ldots,\theta_{F_{R_{n_q}}}) \\
&\in \rea^w.
\end{split}
\end{equation}
We are in position to present the following result, whose proof may be found in \cite{ROMORTBOB}.
\begin{proposition}\label{pro1}
System \eqref{elsys} satisfies the \textit{LRE}:
\begin{equation}\label{newlreel}
y = \boldsymbol\Omega^\top\boldsymbol\theta,
\end{equation}
with $\boldsymbol\theta$ defined in \eqref{thesthe}, where
\begin{equation}\label{eqn:ydot}
    \dot{y} = -\lambda y + \dot{\bm{q}}^\top\boldsymbol\tau,
\end{equation}
with $\lambda>0$ a design parameter, and $\boldsymbol\Omega$ a single-column regressor given by:
\begin{equation}\label{eqn:12}
\dot{\boldsymbol\Omega}=-\lambda\boldsymbol\Omega+
\begin{bmatrix}
p\frac{1}{2}\dot{\bm{q}}^TM_1(\bm{q})\dot{\bm{q}}\\
\vdots\\
p\frac{1}{2}\dot{\bm{q}}^TM_l(\bm{q})\dot{\bm{q}}\\
p{U}_1(\bm{q})\\
\vdots\\
p{U}_r(\bm{q})\\
p\dot{q}_1\\
\vdots\\
p\dot{q}_{n_q}
\end{bmatrix}\in\mathbb{R}^{w},
\end{equation}
where $p:=\frac{d}{dt}$.
\end{proposition}

\subsection{Derivation of the scalar linear regression equations}
\lab{subsec32}
%
In this subsection we apply the first stage of the LS+DR\-EM procedure  \cite{ORTROMARA} to generate, from the power-balance equation parametrization (PBEP) given in \eqref{newlreel}, a new set of {\em scalar} LRE. To be able to  estimate the parameters of the LRE \eqref{newlreel} we impose the {\em necessary} assumption that it is {\em identifiable}  \cite{GOOSINbook}. That is, that there exists a set of time instants---$\{t_i\},\;t_{i+1}>t_i, i \in \bar w$, such that
$$
\rank\Big\{\begmat{\boldsymbol\Omega(t_1)|\boldsymbol\Omega(t_2)|&\cdots&|\boldsymbol\Omega(t_w)}\Big\}=w.
$$

We recall the following result of \cite{WANORTBOB}.

\begin{lemma}
\lab{lem1}
The LRE \eqref{newlreel} is identifiable {\em if and only if} the regressor vector $\boldsymbol\Omega(t)$ is interval exciting (IE)  \cite{KRERIE}. That is,  {there exist} constants $c_c>0$ and $t_c>0$ such that
$$
\int_0^{t_c} \boldsymbol\Omega(s) \boldsymbol\Omega^\top(s)  ds \ge c_c I_w.
$$
\qed
\end{lemma}

Consequently, it is necessary to impose the following.

\begin{assumption}
\lab{ass1}
The regressor $\boldsymbol\Omega(t)$ of Proposition \ref{pro1} is IE.
\end{assumption}

We are in position to state the following result, whose proof may be found in Proposition 1 of   \cite{ORTROMARA}.

\begin{proposition}
\lab{pro2}
		Consider the LRE \eqref{newlreel} verifying Assumption \ref{ass1}. Define the following dynamic extension\footnote{This is part of the LS+DREM interlaced estimator with time-varying forgetting factor  \cite{ORTROMARA}.}
		\begin{align}
\label{LSD1}
			\dot{\hat{\boldsymbol\mu}}   & =\alpha F   \boldsymbol\Omega   (y  -\boldsymbol\Omega^\top    \hat{\boldsymbol\mu}   ),\; \hat{\boldsymbol\mu}  (0)=:\boldsymbol{\mu}_0  \in \rea^w,\\
		\dot {F}  & = -\alpha F  \boldsymbol\Omega \boldsymbol\Omega^\top    F + \beta F ,\;F(0)={1 \over f_0} I_w,\\
\label{LSD3}
		\dot z&=-\beta z,\;z(0)=1,
\end{align}
with
$$
\beta:=\beta_{0}\Big(1-{\|F\| \over \rho}\Big),
$$
and tuning gains $\alpha>0$, $f_{0}>0,\; {\beta_{0} > 0}$ and $\rho \geq {1 \over f_{0}}$.  Then, the new {\em scalar} LREs
\begin{equation}
\label{nLRE}
\caly_i(t) = \Delta(t) \theta_i,\;i \in \bar w,
\end{equation}
holds, where we defined
		\begin{align}
             \label{eqn:Delta}
			\Delta  & :=\det\{I_w- z  f_{0}F\},\\
             \label{eqn:Caly}
			\caly  & := \adj\{I_w- z  f_{0}F\}[\hat{\boldsymbol\mu} - z f_{0} F \boldsymbol{\mu}_0 ].
			\end{align}
Moreover, there exists $\Delta_{\min}>0$ such that
\begequ
\lab{del}
\Delta(t)\geq \Delta_{\min},\;\forall t\geq t_c,
\endequ
hence, {\em $\Delta(t)$ is PE}.
\qed
\end{proposition}
It is clear that, given the PE property of $\Delta(t)$, a simple gradient algorithm ensures the global exponential convergence of the estimate of the parameters $\boldsymbol\theta$, which may be used in a classical indirect adaptive controller. However, as we show below, stronger stability properties are obtained if we use instead a composite adaptive controller.
\section{Three New Composite Adaptive Controllers}
\lab{sec4}
%
In this section we present the main results of the paper, which consists of the design of three new composite adaptive trajectory tracking controllers for the EL system \eqref{robdyn}---all of them relying on the use of the scalar LREs \eqref{nLRE}. The first scheme is a  modified Slotine-Li controller, while the second one is a totally new  PID-like adaptive tracking controller that has the interesting feature of avoiding the filtering operations \cite{SLOLItac} but operates, instead, on the actual tracking error signals. Additionally, the third one contains a modification of the second one, so that exponential convergence can be ensured. For one controller we prove {\em global convergence} of both, the parameter and the tracking errors under the weakest IE assumption. Moreover, for two controllers we are able to prove that the convergence is {\em exponential}.

Throughout the paper we make the standard assumption that the desired trajectory $\bm{q}_\star(t) \in \rea^{n_q}$, as well as its first and second order derivatives, are bounded.
\subsection{A revisited  Slotine-Li controller \cite{SLOLItac}}
\lab{subsec41}
%
The first proposal is a modification of the well-known PD plus adaptive compensation controller of \cite{SLOLItac}. In particular, we show that a key modification of the original adaptive law, which changes the direct adaptive controller into a composite adaptive controller, strengthens the stability conclusions for all closed-loop system states. This modification replaces the original excitation requirement, namely the PE condition on the regressor, with a relaxed IE condition. Following the work of Slotine and Li \cite{SLOLItac}, we define the auxiliary signals
\begin{equation}
	\label{SLerrors}
\bm{s} := \dot {\tilde{\bm{q}}} + K_{S} \tilde{\bm{q}},  \quad  \dot{\bm{q}}_{r} := \dot{\bm{q}}_\star  - K_{S} \tilde{\bm{q}},
\end{equation}
where $K_S \in \rea^{n_q \times n_q}$ is {\em diagonal, positive definite} matrix.

We are now in position to present our first main result in the following Proposition.
\begin{proposition}\label{pro3}
Consider the EL system \eqref{elsys}  satisfying  the LRE \eqref{newlreel} and the  new scalar LREs \eqref{nLRE}. Define the control signal via the following composite adaptive scheme
\begin{align}
\label{newt0}
\boldsymbol\tau=& {k_I} Y(\bm{q},\dot{\bm{q}},\dot{\bm{q}}_r, \ddot{\bm{q}}_r) \hat \theta- K_{D}\bm{s} -K_{P} \tilde{\bm{q}} , \\
\label{newt00}
\dot {\hat{\boldsymbol\theta}}=&  -{k_I} \Gamma Y(\bm{q},\dot{\bm{q}},\dot{\bm{q}}_r, \ddot{\bm{q}}_r)^\top \bm{s} + \Gamma \Delta (\caly- {k_I} \Delta \hat{\boldsymbol\theta}),
\end{align}
with {\em diagonal, positive definite} matrices $K_P$, $K_D$, $\Gamma \in \rea^{n_q \times n_q}$, $k_I$ being a positive {\em scalar}, $\caly$ given by \eqref{eqn:Caly}, $\Delta$ given by \eqref{eqn:Delta}, and $Y(\cdot)$ obtained from \eqref{eqn:phi2}. Suppose the regressor $\boldsymbol\Omega(t)$ verifies Assumption \ref{ass1}.  Then, the origin of the resultant closed-loop system, that is,
 \begin{equation*}
 \begin{bmatrix}\bm{s}^T & {\tilde{\bm{q}}}^T & \tilde{\boldsymbol\theta}^T\end{bmatrix}=\begin{bmatrix}\bm{0}_n^T&\bm{0}_n^T&\bm{0}_w^T \end{bmatrix}^T,
  \end{equation*}
  is a globally {\em exponentially} stable (GES) equilibrium point, where $\tilde{\bm{q}}:=\bm{q}-\bm{q}_\star$ and  ${\tilde{\boldsymbol\theta}}:=\hat{\boldsymbol\theta}-{1 \over k_I}\boldsymbol\theta$.
\end{proposition}
 \begin{proof}
 The proof of Proposition \ref{pro3} follows directly the procedure used in the proof of Proposition 3 in \cite{ROMORT}, considering the proposed extended parametrization given in \eqref{newlreel}. Therefore, it is omitted for brevity.

\end{proof}
\subsection{A new PID-like adaptive global tracking controller}
\lab{subsec42}
As seen in \eqref{newt0}, the regressor vector $Y(\cdot)$ in the Slotine--Li controller depends on the auxiliary signals $\dot{\bm{q}}_r$ and $\ddot{\bm{q}}_r$, which, as shown in \cite{BERORTNIJ}, may induce parameter drift in the presence of noise. To overcome this problem, in this subsection we propose a simple {\em adaptive PID-like controller} whose regressor depends on the {\it real} tracking errors. To this end, invoking Property \ref{prop1}, we define a \emph{new} parametrization of the system dynamics:
\begin{equation}\label{Yn}
Y_n(\bm{q},\dot{\bm{q}}, \dot{\bm{q}}_\star, \ddot{\bm{q}}_\star) \theta= M(\bm{q})\ddot{\bm{q}}_\star  +C(\bm{q}, \dot{\bm{q}})\dot{\bm{q}}_\star +\nabla U(\bm{q})+\mathcal{R}\dot{\bm{q}},
\end{equation}
where $Y_n:\mathbb{R}^{n_q}\times\mathbb{R}^{n_q}\times \mathbb{R}^{n_q} \times \mathbb{R}^{n_q}  \mapsto \mathbb{R}^{n_q \times w}$ and $\boldsymbol\theta$ are the unknown model parameters defined in \eqref{thesthe}.
We are now in a position to present the second main result of the paper. Compared with the proposal in \cite[Proposition 4]{ROMORT}, we include an extended PBEP given by \eqref{newlreel} to account for \emph{unknown} friction-related parameters and prove the convergence to zero of the joint position error $\tilde{\bm{q}}$ by invoking Matrosov's theorem; this result was not included in \cite{ROMORT}.
\begin{proposition}
\label{pro4}
Consider the EL system \eqref{elsys}, satisfying the LRE \eqref{newlreel} and the new scalar LREs \eqref{nLRE}. Define the control signal via the following composite adaptive PID-like scheme
\begin{align}
\label{newt1}
\boldsymbol\tau=&  -K_{P} \tilde{\bm{q}}+{k_I}Y_n(\bm{q},\dot{\bm{q}}, \dot{\bm{q}}_\star , \ddot{\bm{q}}_\star) \hat{\boldsymbol\theta}-K_{D} \dot {\tilde{\bm{q}}} , \\
\label{newt11}
\dot {\hat{\boldsymbol\theta}}=&  -{k_I}\Gamma Y_n(\bm{q},\dot{\bm{q}}, \dot{\bm{q}}_\star , \ddot{\bm{q}}_\star)^\top \dot {\tilde{\bm{q}}} +\Gamma \Delta (\caly- {k_I}\Delta \hat{\boldsymbol\theta}),
\end{align}
with positive definite adaptive gains $K_P$, $K_D$, $\Gamma$ $\in \rea^{n_q \times n_q}$, $k_I$ is a positive {\em scalar}, the matrix $Y_n$ of equation \eqref{Yn}, $\caly$ given by \eqref{eqn:Caly} and $\Delta$ of \eqref{eqn:Delta}. Suppose the regressor $\boldsymbol\Omega(t)$ verifies Assumption \ref{ass1}. Then, the origin of the closed-loop system, that is,
 \begin{equation*}
 \begin{bmatrix}\tilde{\bm{q}}^T & \dot{\tilde{\bm{q}}}^T & \tilde{\boldsymbol\theta}^T\end{bmatrix}=\begin{bmatrix}\bm{0}_n^T&\bm{0}_n^T&\bm{0}_w^T \end{bmatrix}^T,
  \end{equation*}
   is a globally uniformly asymptotically stable (GUAS) equilibrium point, that is,
\begin{equation}\label{eqn:resultprop4}
\centering
\lim_{t \to \infty}
\begin{bmatrix}
\tilde{\bm{q}}(t) \\[1mm]
\dot{\tilde{\bm{q}}}(t) \\[1mm]
\tilde{\boldsymbol\theta}(t)
\end{bmatrix}
= \bm{0}_{2n+w}.
\end{equation}
\end{proposition}
\begin{proof}
Rewriting the EL dynamics \eqref{robdyn} in terms of $\tilde{\bm{q}}$ and $\dot{\tilde{\bm{q}}}$, and substituting the proposed control law \eqref{newt1}--\eqref{newt11}, yields the error dynamics:
\begin{align}
M(\bm{q})\ddot {\tilde{\bm{q}}} &=- [C(\bm{q},\dot{\bm{q}})+K_D]\dot{\tilde{\bm{q}}} -K_P\tilde{\bm{q}} +k_I Y_n \tilde{ \boldsymbol\theta}, \label{dyn1} \\
\dot {\tilde{\boldsymbol\theta}} &= -k_I \Gamma Y_n^\top \dot{\tilde{\bm{q}}} -k_I\Gamma \Delta^2 \tilde{\boldsymbol\theta}.
\label{dyn2}
\end{align}
To analyze the behavior of the closed-loop system \eqref{dyn1}--\eqref{dyn2}, we propose the following Lyapunov function:
\begin{equation}
\mathbb{W}(\tilde{\bm{q}},\dot{\tilde{\bm{q}}},\tilde{\boldsymbol\theta})=\frac{1}{2}\dot {\tilde{\bm{q}}}^\top M(\bm{q})\dot {\tilde {\bm{q}}} +\frac{1}{2} \tilde{\bm{q}}^\top K_{P} \tilde{\bm{q}} +\frac{1}{2}\tilde{\boldsymbol\theta}^\top\Gamma^{-1} \tilde{\boldsymbol\theta},
\label{defW}
\end{equation}
whose time derivative yields
 \begin{align}
\dot {\mathbb{W}}=&\dot {\tilde{\bm{q}}}^\top M\ddot {\tilde{\bm{q}}} +\frac{1}{2}\dot {\tilde{\bm{q}}}^\top \dot M\dot {\tilde{\bm{q}}} +\tilde{\boldsymbol\theta}^\top \Gamma^{-1} \dot {\tilde{\boldsymbol\theta}} +\tilde{\bm{q}}^\top K_{P} \dot {\tilde{\bm{q}}},\nonumber \\
=& -\dot{\tilde{\bm{q}}}^\top K_{D} \dot {\tilde{\bm{q}}} -{k_I} \Delta^2 \|\tilde{\boldsymbol\theta}\|^2,  \label{dotU11}  \\
\leq & -\dot{\tilde{\bm{q}}}^T K_{D} \dot{\tilde{\bm{q}}} -k_I \Delta_{\min}
 \| \tilde{\boldsymbol\theta} \|^2 , \; \mbox{for} \; t \geq  t_c, \label{dotU1}
  \end{align}
where Property~\ref{pp:1} has been invoked. From \eqref{defW} and \eqref{dotU1}, we conclude that the origin of the closed-loop system \eqref{dyn1}--\eqref{dyn2} is globally uniformly stable and that $\tilde{\bm{q}}$, $\dot{\tilde{\bm{q}}}$, and $\tilde{\boldsymbol\theta}$ are bounded. By invoking the vectorial form of Barbalat's Lemma\footnote{See \cite{ControlOfRobots}, Lemma A.5.}, we can conclude that $\dot{\tilde{\bm{q}}}$ and $\tilde{\boldsymbol\theta}$ converge to zero; however, no such conclusion can be directly drawn for $\tilde{\bm{q}}$. To establish convergence of all states of the closed-loop system \eqref{dyn1}--\eqref{dyn2} to zero, we employ Matrosov's theorem. The technical details are provided in the Appendix. This completes the proof of Proposition \ref{pro4}.
 \end{proof}
We make the important observation that, in contrast with the scheme of Proposition \ref{pro3}, the latter controller ensures only global asymptotic stability, but not exponential stability, of the equilibrium point of its corresponding closed-loop system. Moreover, the fact that the former controller performs better is clearly illustrated by the experimental results in the next section.

\subsection{A new GES controller}
Motivated by the PID-like controller of Proposition~\ref{pro4}, we now present a slight modification of that controller, specifically in the adaptation law, so that the resulting closed-loop system admits a strict Lyapunov function. This modification strengthens the stability properties, yielding exponential convergence of the tracking and parameter errors to zero. Compared again with the proposal of Proposition~\ref{pro3}, the proposed controller avoids the use of the system dynamics regressor $Y(\bm{q},\dot{\bm{q}},\dot{\bm{q}}_r,\ddot{\bm{q}}_r)$, which, as explained before, is related to parameter drift phenomena and noise sensitivity.
\begin{proposition}
\label{pro5}
Consider the EL system \eqref{elsys}  satisfying  the LRE \eqref{newlreel} and the  new scalar LREs \eqref{nLRE}. Define the control signal via the following composite adaptive scheme
\begin{align}
\label{newt21}
\boldsymbol\tau=&  -K_{P} \tilde{\bm{q}}+{k_I}Y_n(\bm{q},\dot{\bm{q}}, \dot{\bm{q}}_\star , \ddot{\bm{q}}_\star) \hat{\boldsymbol\theta}-K_{D} \dot {\tilde{\bm{q}}} , \\
\nonumber
\dot {\hat{\boldsymbol\theta}}=&  -{k_I}\Gamma Y_n(\bm{q},\dot{\bm{q}}, \dot{\bm{q}}_\star , \ddot{\bm{q}}_\star)^\top \Bigg[ \beta\dot{\tilde{\bm{q}}}+\frac{2\tilde{\bm{q}}}{1+2\| \tilde{\bm{q}} \|^2} \Bigg]\\
\label{newt31}
 &+\Gamma \Delta (\caly- {k_I}\Delta \hat{\boldsymbol\theta}),
\end{align}
with positive definite adaptive gains $K_P$, $K_D$, $\Gamma$ $\in \rea^{n_q \times n_q}$, the matrix $Y_n$ of equation \eqref{Yn}, $\caly$ given by \eqref{eqn:Caly}, $\Delta$ of \eqref{eqn:Delta},$k_I$ is a positive {\em scalar} and $\beta$ is such that
\begin{equation}\label{eqn:beta}
\beta>\max\{\beta_{1},\beta_{2},\beta_{3}\},
\end{equation}
where
\begin{eqnarray*}
\beta_{1}&=&\frac{2\lambda_{\max}\{M\}}{\sqrt{\lambda_{\min}\{M\}\lambda_{\min}\{ K_p \} }},\\
\beta_{2}&=&2\sqrt{\frac{\lambda_{\max}\{M\}}{\lambda_{\max}\{ K_p \}}},\\
\beta_{3}&=&\frac{1}{\lambda_{\min}\{K_D\}}\Bigg[ \frac{(\lambda_{\max}\{K_D\}+k_C\|\dot{\bm{q}}_\star(t)\|)^2}{2\lambda_{\min}\{K_p\}} \\
&&+4\lambda_{\max}\{M\}+\frac{k_C}{\sqrt{2}}\Bigg].
\end{eqnarray*}
Suppose the regressor $\boldsymbol\Omega(t)$ verifies Assumption \ref{ass1}.  Then, the origin of the resultant closed-loop system, that is,
 \begin{equation*}
 \begin{bmatrix}\tilde{\bm{q}}^T & \dot{\tilde{\bm{q}}}^T & \tilde{\boldsymbol\theta}^T\end{bmatrix}=\begin{bmatrix}\bm{0}_n^T&\bm{0}_n^T&\bm{0}_w^T \end{bmatrix}^T,
  \end{equation*}
  is a GES equilibrium point ensuring the existence of constants $\rho_4>0$ and $\rho_5>0$ such that
\begin{equation}\label{eqn:resultprop33}
\left  \|\begmat{ \tilde{\bm{q}}(t) \\ \dot{\tilde{\bm{q}}}(t) \\ {\tilde{\boldsymbol\theta}}(t)}\right \| \leq \rho_4e^{-\rho_5 t} \left \| \begmat{\tilde{\bm{q}}(0) \\ \dot{\tilde{\bm{q}}}(0) \\ {\tilde{\boldsymbol\theta}}(0)}\right\|,\; \forall t \geq 0.
\end{equation}
\end{proposition}
\begin{proof}
The closed-loop system is obtained by rewriting the EL dynamics \eqref{robdyn} in terms of $\tilde{\bm{q}}$ and $\dot{\tilde{\bm{q}}}$ and substituting the proposed control law \eqref{newt21}--\eqref{newt31}, which yields:
\begin{align}
\ddot {\tilde{\bm{q}}} &=-M(\bm{q})^{-1} [C(\bm{q},\dot{\bm{q}})+K_D]\dot{\tilde{\bm{q}}} -K_P\tilde{\bm{q}} +k_I Y_n \tilde{ \boldsymbol\theta}, \label{dyn21} \\
\dot {\tilde{\boldsymbol\theta}} &= -k_I \Gamma Y_n^\top \Bigg[ \beta\dot{\tilde{\bm{q}}}+\frac{2\tilde{\bm{q}}}{1+2\| \tilde{\bm{q}} \|^2} \Bigg] -k_I\Gamma \Delta^2 \tilde{\boldsymbol\theta}.
\label{dyn22}
\end{align}
The proof that the origin of \eqref{dyn21}--\eqref{dyn22} is a GES equilibrium point is carried out in two steps. First, we define the following Lyapunov function candidate:
\begin{equation}
\begin{split}
V(\tilde{\bm{q}},\dot{\tilde{\bm{q}}},\tilde{\boldsymbol\theta})=&\frac{\beta}{2}\dot {\tilde{\bm{q}}}^\top M(\bm{q})\dot {\tilde {\bm{q}}} +\frac{\beta}{2} \tilde{\bm{q}}^\top K_{P} \tilde{\bm{q}} +\frac{1}{2}\tilde{\boldsymbol\theta}^\top\Gamma^{-1} \tilde{\boldsymbol\theta}\\
&+\frac{2\dot{\tilde{\bm{q}}}M(\bm{q})\tilde{\bm{q}}}{1+2\tilde{\bm{q}}^T\tilde{\bm{q}}},
\end{split}
\label{defVS}
\end{equation}
It will be proved that $V$ is a positive definite, radially unbounded, and decrescent function. To that end, observe that the fourth added of the right-hand side in \eqref{defVS} satisfies
\begin{equation}
\frac{2\dot{\tilde{\bm{q}}}M(\bm{q})\tilde{\bm{q}}}{1+2\tilde{\bm{q}}^T\tilde{\bm{q}}}\geq -2\lambda_{\max}\{M(\bm{q})\}\|\dot{\tilde{\bm{q}}}\| \|{\tilde{\bm{q}}}\|.
\end{equation}
Hence, the Lyapunov function candidate \eqref{defVS} verifies the inequality
\begin{equation}\label{Ve1}
V(\boldsymbol\eta)\geq \boldsymbol\vartheta^T E_{\min}\boldsymbol\vartheta,
\end{equation}
where $\boldsymbol\eta=\begin{bmatrix}\tilde{\bm{q}}^T& \dot{\tilde{\bm{q}}}^T&\tilde{\boldsymbol\theta}^T\end{bmatrix}^T,\boldsymbol\vartheta=\boldsymbol\vartheta(\boldsymbol\eta)=\begin{bmatrix}\|\tilde{\bm{q}}\| &  \|\dot{\tilde{\bm{q}}}\| & \|\tilde{\boldsymbol\theta}\| \end{bmatrix}^T$ with the matrix $E_{\min}$ defined by:
\begin{equation}
\begin{split}
&E_{\min}= \\
&\begin{bmatrix}
\frac{\beta}{2}\lambda_{\min}\{K_p\} & -\lambda_{\max}\{M(\bm{q})\} & 0\\
-\lambda_{\max}\{M(\bm{q})\} & \frac{\beta}{2}\lambda_{\min}\{M(\bm{q})\} & 0\\
0 & 0 & \frac{1}{2}\lambda_{\min}\{\Gamma^{-1}\}
\end{bmatrix}.
\end{split}
\end{equation}
Since $K_p$ and $M$ are positive definite matrices and $\beta$ is selected to satisfy \eqref{eqn:beta} then it follows that $V$ of \eqref{defVS} is a positive definite function, and moreover, it is a radially unbounded one. Thus, $V$ also satisfies:
\begin{equation}\label{Ve2}
V(\boldsymbol\eta)\leq \boldsymbol\vartheta^T E_{\max}\boldsymbol\vartheta,\\
\end{equation}
where
\begin{equation}
\begin{split}
&E_{\max}=\\
& \begin{bmatrix}
\frac{\beta}{2}\lambda_{\max}\{K_p\} & \lambda_{\max}\{M(\bm{q})\} & 0\\
\lambda_{\max}\{M(\bm{q})\} & \frac{\beta}{2}\lambda_{\max}\{M(\bm{q})\} & 0\\
0 & 0 & \frac{1}{2}\lambda_{\max}\{\Gamma^{-1}\}
\end{bmatrix}.
\end{split}
\end{equation}
Hence, $V$ is a decrescent function, since $\lambda_{\max}\{K_p\}$ and $\lambda_{\max}\{M(\bm{q})\}$ are strictly positive constants, and $\beta$ satisfies \eqref{eqn:beta}.
\par On the other hand, the time derivative of the Lyapunov function candidate \eqref{defVS} along the trajectories of the closed-loop system \eqref{dyn21}--\eqref{dyn22} yields:
 \begin{align}
\dot {V}=&-\beta\dot{\tilde{\bm{q}}}^TM(\bm{q})\ddot{\tilde{\bm{q}}}+\frac{\beta}{2}\dot{\tilde{\bm{q}}}^T\dot{M}(\bm{q})\dot{\tilde{\bm{q}}}+\beta\tilde{\bm{q}}^TK_P\dot{\tilde{\bm{q}}} \nonumber \\
&+\tilde{\boldsymbol\theta}\Gamma^{-1}\dot{\tilde{\boldsymbol\theta}}+\frac{d}{dt}\Bigg[ \frac{2\dot{\tilde{\bm{q}}}^TM(\bm{q})\tilde{\bm{q}}}{1+2\|\tilde{\bm{q}}\|^2}\Bigg],\\
=&-\beta\dot{\tilde{\bm{q}}}^TK_D\dot{\tilde{\bm{q}}}-\frac{2\tilde{\bm{q}}^TK_P\tilde{\bm{q}}}{1+2\|\tilde{\bm{q}}\|^2}-\frac{8\dot{\tilde{\bm{q}}}^TM(\bm{q})\tilde{\bm{q}}\tilde{\bm{q}}^T\dot{\tilde{\bm{q}}}}{(1+2\|\tilde{\bm{q}}\|^2)^2} \nonumber \\
&-\frac{2\dot{\tilde{\bm{q}}}^TK_D^T\tilde{\bm{q}}}{1+2\|\tilde{\bm{q}}\|^2}+\frac{2\dot{\tilde{\bm{q}}}^TC(\bm{q},\dot{\bm{q}})\tilde{\bm{q}}}{1+2\|\tilde{\bm{q}}\|^2}
+\frac{2\dot{\tilde{\bm{q}}}^TM(\bm{q})\dot{\tilde{\bm{q}}}}{1+2\|\tilde{\bm{q}}\|^2} \nonumber \\
&-\Delta^2\|\boldsymbol\theta\|^2,
 \label{dotV1}
  \end{align}
 where Property~\ref{pp:1} has been invoked. Using the following inequalities—obtained from Properties~\ref{pp:2}–\ref{pp:4}—:
\begin{align}
-\beta\dot{\tilde{\bm{q}}}^TK_D\dot{\tilde{\bm{q}}}&\leq -\beta\lambda_{\min}\{K_D\}\|\dot{\tilde{\bm{q}}}\|^2,\\
-\frac{2\tilde{\bm{q}}^TK_P\tilde{\bm{q}}}{1+2\|\tilde{\bm{q}}\|^2}&\leq -\frac{2\lambda_{\min}\{K_P\}\| \tilde{\bm{q}} \|^2}{1+2\| \tilde{\bm{q}}\|^2},\\
-\frac{8\dot{\tilde{\bm{q}}}^TM(\bm{q})\tilde{\bm{q}}\tilde{\bm{q}}^T\dot{\tilde{\bm{q}}}}{(1+2\|\tilde{\bm{q}}\|^2)^2}&\leq\frac{8\|\dot{\tilde{\bm{q}}}\|^2\|{\tilde{\bm{q}}}\|^2\lambda_{\max}\{M(\bm{q})\}}{(1+2\|{\tilde{\bm{q}}}\|^2)^2}\\
&\leq 2\|\dot{\tilde{\bm{q}}}\|^2\lambda_{\max}\{M(\bm{q})\},\\
-\frac{2\dot{\tilde{\bm{q}}}^TK_D^T\tilde{\bm{q}}}{1+2\| \tilde{\bm{q}} \|^2}&\leq\frac{2\lambda_{\max}\{K_D\}\|\dot{\tilde{\bm{q}}}\|\|\tilde{\bm{q}}\|}{1+2\| \tilde{\bm{q}} \|^2}\\
\nonumber
\frac{2\dot{\tilde{\bm{q}}}^TC(\bm{q},\dot{\bm{q}})\tilde{\bm{q}}}{1+2\|\tilde{\bm{q}}\|^2}&\leq \frac{2\|\dot{\tilde{\bm{q}}}\|\|\tilde{\bm{q}}\|}{1+2\|\tilde{\bm{q}}\|^2}\Bigg[k_C\|\dot{\tilde{\bm{q}}}\|+k_C\|\dot{\bm{q}}_\star(t)\| \Bigg]\\
&\leq \frac{k_C}{\sqrt{2}}\|\dot{\tilde{\bm{q}}}\|^2+\frac{2k_C\|\dot{\bm{q}}_\star(t)\|\|\tilde{\bm{q}}\|\|\dot{\tilde{\bm{q}}}\|}{1+2\|\tilde{\bm{q}}\|^2},\\
\frac{2\dot{\tilde{\bm{q}}}^TM(\bm{q})\dot{\tilde{\bm{q}}}}{1+2\|\tilde{\bm{q}}\|^2}&\leq2\lambda_{\max}\{M(\bm{q})\}\| \dot{\tilde{\bm{q}}} \|^2,\\
-\Delta^2\|\boldsymbol\theta\|^2&\leq -\Delta_{\min}\| \tilde{\boldsymbol\theta} \|^2,\;\mbox{for}\;t\geq t_c,
\end{align}
the time derivative of the Lyapunov function candidate \eqref{defVS} takes the form
\begin{equation}\label{eqn:vdot}
\dot{V}\leq - \boldsymbol\vartheta^TA\boldsymbol\vartheta,
\end{equation}
where
\begin{equation}\label{eqn:A}
A=
\begin{bmatrix}
a_{11}& a_{12} & 0\\
a_{21} & a_{22}&0\\
0&0&a_{33}
\end{bmatrix}
\end{equation}
with
\begin{align}
a_{11}&=\frac{2\lambda_{\min}\{K_P\}}{1+2\|\tilde{\bm{q}}\|^2},\\
a_{12}&=- \frac{\lambda_{\max}\{K_D\}-k_C\|\dot{\bm{q}}_\star\|}{1+2\|\tilde{\bm{q}}(t)\|^2},\\
a_{21}&=a_{12},\\
a_{22}&=\beta\lambda_{\min}\{K_D\}-4\lambda_{\max}\{M(\bm{q})\}-\frac{k_C}{\sqrt{2}},\\
a_{33}&=\Delta_{\min}.
\end{align}
Therefore, the time derivative in \eqref{eqn:vdot} results in a globally negative definite function because $\beta$ satisfies \eqref{eqn:beta}, which, in turn, implies that matrix $A$ in \eqref{eqn:A} is a positive definite one.
\par Hence, from \eqref{Ve1}, \eqref{Ve2} and \eqref{eqn:vdot}, together with the fact that $\beta$ is selected such that it verifies \eqref{eqn:beta}; and consequently  $A$, $E_{\min}$, and $E_{\max}$ are positive definite. Thus, the Lyapunov function \eqref{defVS} satisfies
\begin{align}
	\label{eq:th2a}
	c_1 \|\boldsymbol{\eta}\|^2 \leq &\, V(t,\boldsymbol{\eta}) \leq  c_2 \|\boldsymbol{\eta}\|^2, \\
	\label{eq:th2b}
	\frac{\partial V}{\partial t} + \frac{\partial V}{\partial \boldsymbol{\eta}}(t, \boldsymbol{\eta})  \leq &\, -c_3 \|\boldsymbol{\eta}\|^2,
\end{align}
where $c_1=\lambda_{\min}\{E_{\min}\}$, $c_2=\lambda_{\max}\{E_{\max}\}$, and $c_3=\lambda_{\min}\{A\}$. Since $c_1$, $c_2$, and $c_3$ are strictly positive constants, by invoking Theorem~4.10 (pp.~154–155) of \cite{khalil2002nonlinear}, global exponential stability of the origin of the closed-loop system \eqref{dyn22} is established; therefore, \eqref{eqn:resultprop33} holds. This  completes the proof of Proposition~\ref{pro5}.
\end{proof}
 \begin{remark}
Note that the main difference between Proposition~\ref{pro4} and Proposition~\ref{pro5} lies {\it solely} in the adaptation law. Specifically, the difference between \eqref{newt11} and \eqref{newt31} is the additional second term on the right-hand side of \eqref{newt31}. This term is necessary to cancel certain components associated with the time derivative of the cross term in the Lyapunov function \eqref{defVS}, which is necessary for strengthening the stability conclusions.
 \end{remark}
\subsection{Tuning guidelines for controller gains}\label{sub:Gains}
 The proposed controllers of Proposition \ref{pro3}-\ref{pro5} include: a proportional term $K_P\tilde{\bm{q}}$ acting on the joint position tracking error, a derivative term $K_D\dot{\tilde{\bm{q}}}$ acting on the joint velocity error, and an integral/adaptive term shaped by the scalar gain $k_I$ multiplying the estimated-parameter contribution. Therefore, classical PID-tuning intuition can be used as a starting point: $K_P$ primarily stiffens the position error dynamics, $K_D$ increases damping and reduces overshoot, and $k_I$ governs the aggressiveness of the adaptive/integral action (faster parameter adaptation versus noise sensitivity) see for instance \cite{yu}. Regarding the LS$+$DREM estimator, the estimator gains can be selected following standard guidelines for least-squares schemes with forgetting factor, as discussed in adaptive control textbooks (see, e.g., \cite{SASBODbook}).
 \begin{figure}[htp!]
\centering
\subfloat[Picture of the robot.]{%
\resizebox*{5cm}{!}{\includegraphics{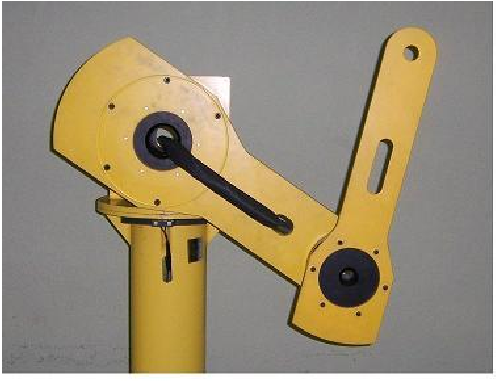}}
}\hspace{5pt}
\subfloat[Sketch of the robot arm.]{%
\resizebox*{5cm}{!}{\includegraphics{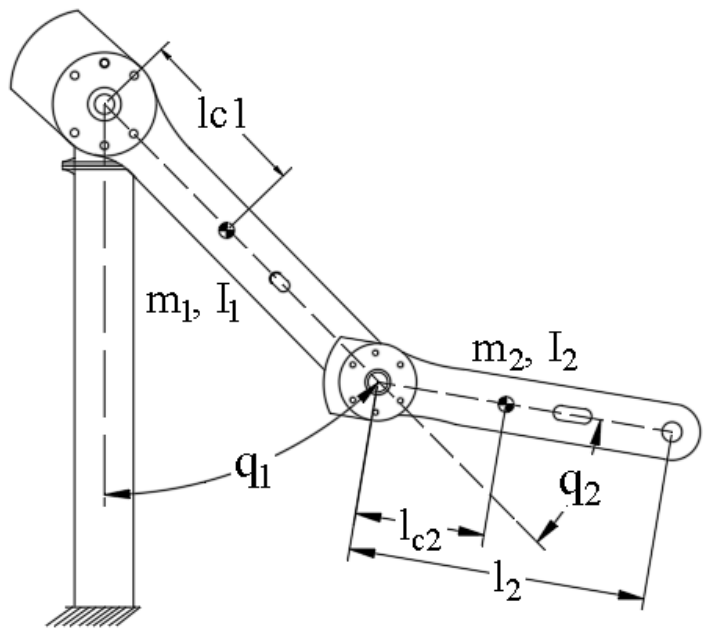}}
}
\caption{Two degrees-of-freedom manipulator arm.} \label{Fig:1to2}
\end{figure}
\section{Experimental Evaluation}\label{secc:ExperimentalEvaluation}
\begin{table*}[h!]
	\caption{Physical parameters of the robot manipulator arm.}
	\label{tab:parrob}
	\begin{center}
		\begin{tabular}{cccc} \hline
			Description  & Notation & Value & Units \\ \hline
			Length of Link 1 & $l_1$ & 0.45 & $\mathrm{m}$ \\
			Length of Link 2 & $l_2$ & 0.45 & $\mathrm{m}$ \\
			Dist. to the com ($l_1$) & $l_{c_1}$ & 0.091 & $\mathrm{m}$  \\
			Dist. to the com ($l_2$) & $l_{c_2}$ & 0.048 & $\mathrm{m}$  \\
			Mass of Link 1 & $m_1$ & 23.902 & $\mathrm{kg}$ \\
			Mass of Link 2 & $m_2$ & 3.88 & $\mathrm{kg}$ \\
			Inert. rel. to com ($l_1$) & $I_1$ & 1.266  & $\mathrm{\frac{kg \ m^2}{rad^2}}$ \\
			Inert. rel. to com ($l_2$) & $I_2$ & 0.093  & $\mathrm{\frac{kg \ m^2}{rad^2}}$ \\
			Gravity acceleration & $g$ & 9.81 & $\mathrm{\frac{m}{s^2}}$ \\
            Viscous coef. rel. to link 1 & $f_{v_1}$ & 2.288 &  $\mathrm{Nm \ s/rad}$ \\
            Viscous coef. rel. to link 2 & $f_{v_2}$ & 0.175 &  $\mathrm{Nm \ s/rad}$ \\ \hline
		\end{tabular}
	\end{center}
\end{table*}

\begin{table}[h!]
	\caption{Dynamic parameters and its units based on the power balance equation parametrization.}
	\label{tab:parunits}
	\begin{center}
		\begin{tabular}{ccc} \hline
			    Parameter        & Value & Units \\
        \hline
			$\theta_1$  & $2.351$ & $\mathrm{\frac{kg \ m^2}{rad^2}}$ \\
			$\theta_2$  & $0.083$ & $\mathrm{\frac{kg \ m^2}{rad^2}}$ \\
			$\theta_3$  & $0.101$ & $\mathrm{\frac{kg \ m^2}{rad^2}}$ \\
			$\theta_4$  & $3.921$ & $\mathrm{kg \  m}$\\
			$\theta_5$  & $0.186$ & $\mathrm{kg \  m}$ \\
            $\theta_6$  & $2.288$ & $\mathrm{Nm \ s/rad}$ \\
            $\theta_7$  & $0.175$ & $\mathrm{Nm \ s/rad}$ \\
            \hline
		\end{tabular}
	\end{center}
\end{table}
Experimental tests upon a two degree-of-freedom robot manipulator have been carried out in order to illustrate the performance of the proposed adaptive controllers. The robot manipulator is a direct-drive arm with two vertical links and has been designed and built at the CICESE Research Center, Mexico \cite{Cicese41,Cicese42}. A picture of robot is shown in Figure \ref{Fig:1to2}a, and its schematic sketch is depicted in Figure \ref{Fig:1to2}b, where $q_1$ and $q_2$ are the joint positions. The physical parameters of the two links are listed in Table \ref{tab:parrob}, where it is summarized the meaning of the parameters involved and their numerical values, where $l_1$ and $l_2$ mean lengths of Link 1 and Link 2, respectively, `Dist.', `Inert.', `coef.' and `com' stand for Distance, Inertia, coefficient and center of mass, respectively. The vector of torque control inputs is denoted by $\boldsymbol\tau$. The dynamic model \eqref{robdyn} can be described with the matrix
\begin{equation}
M(\bm{q}) = \begin{bmatrix}
\theta_1 + 2\theta_2\cos(q_2) & \theta_3 + \theta_2\cos(q_2) \\
\theta_3 + \theta_2\cos(q_2) & \theta_3
\end{bmatrix},
\end{equation}
and the potential energy function
\begin{equation}
U(\bm{q}) = -g\cos(q_1)\theta_4-g\cos(q_1+q_2)\theta_5.
\end{equation}
The explicit dissipative energy function $F_R$ in \eqref{eqn:F} is given by
\begin{equation}
F_R(\dot{\bm{q}})=\frac{1}{2}\dot{q}_1^2\theta_6+\frac{1}{2}\dot{q}_2^2\theta_7.
\end{equation}
%
 The {\it unknown} constant parameters are defined as:
\begin{eqnarray*}
\theta_1&=&m_1l_{c_1}^2+m_2l_1^2+m_2l_{c_2}^2+I_1+I_2,\\
\theta_2&=&l_1m_2l_{c_2},\\
\theta_3&=&m_2l_{c_2}^2+I_2,\\
\theta_4&=&l_{c_1}m_1+m_2l_1,\\
\theta_5&=&m_2l_{c_2},\\
\theta_6&=&f_{v_1},\\
\theta_7&=&f_{v_2},
\end{eqnarray*}
being
\begin{equation}
\boldsymbol\theta=\mathrm{col}\{\theta_1,\theta_2,\theta_3,\theta_4,\theta_5,\theta_6,\theta_7\}.
\end{equation}
The real values of the dynamic parameters considering the measured data of Table \ref{tab:parrob} and the above definition are presented in Table \ref{tab:parunits}. According to Section \ref{subsec31}, and regarding equation \eqref{parmu}, we define
\begin{equation*}
\begin{split}
M_1 := \begin{bmatrix}
1 & 0 \\
0 & 0
\end{bmatrix},
&M_2(q_2) := \cos(q_2)\begin{bmatrix}
2 & 1 \\
1 & 0
\end{bmatrix},\\
&M_3 := \begin{bmatrix}
0 & 1 \\
1 & 1
\end{bmatrix},
\end{split}
\end{equation*}
\begin{equation*}
U_1(\bm{q}) := -g\cos(q_1), \quad U_2(\bm{q}) := -g\cos(q_1+q_2).
\end{equation*}
From the above definitions, it should be obvious to the reader that $\ell = 3$, $r = 2$, $n_q = 2$ and $w = 7$. Thus, the LRE \eqref{newlreel} holds with $y = H(p)[\dot{\bm{q}}^\top \boldsymbol{\tau}]$, the LTI filter $H(p) = \frac{1}{p+\lambda}$, and the regressor vector
\begin{equation}\label{eqn:omegap}
\boldsymbol \Omega = pH(p)\begin{bmatrix}
\frac{1}{2}\dot{q}_1^2 \\
(\dot{q}_1^2+\dot{q}_1\dot{q}_2)\cos(q_2) \\
\frac{1}{2}q_2^2+\dot{q_1}\dot{q_2} \\
-g\cos(q_1) \\
-g\cos(q_1+q_2)\\
\dot{q}_{1}^2\\
\dot{q}_{2}^2
\end{bmatrix}.
\end{equation}

 For the implementation of the LTI filter $H(p)$, we set $\lambda=1$ in all experiments. In the experimental implementation, the joint velocities were computed from the measured joint positions using the two-point backward numerical differentiation method; that is, for joint $i$,
\begin{equation}
\dot{q}_i[k]\approx \frac{q_i[k]-q_i[k-1]}{T_s},
\end{equation}
where $T_s$ denotes the sampling period. Moreover, for all experiments, the initial positions and velocities were fixed as $\bm{q}(0)=\begin{bmatrix}0&0\end{bmatrix}^T$ and $\dot{\bm{q}}(0)=\begin{bmatrix}0&0\end{bmatrix}^T$ respectively. The experiments were executed on a real-time PC--based platform running in a Windows environment called \emph{WinMechLab} \cite{WinMechLab}. The identification algorithm and controllers of Propositions \ref{pro3}-\ref{pro5} were discretized according to the improved Euler method and was executed at a $2.5 \ \mathrm{[ms]}$ sampling period utilizing the MultiQ-PCI data acquisition board of Quanser Consulting Inc.


Regarding the identification procedure of Section \ref{subsec32}, the corresponding tuning gains were set as indicated in Table \ref{tab:LSDpar}. The spectral norm of the matrix $F$, that is, $\|F\|$, was computed using a power iteration algorithm.
\par Moreover, we used the desired joint positions $\bm{q}_{\star}(t)=\col(q_{1\star}(t), q_{2\star}(t))$ as test signals, as in \cite{ControlOfRobots}:
\begin{align}
	\label{eq:qdma}
	\begin{bmatrix}
		q_{1\star}(t) \\
		q_{2\star}(t)
	\end{bmatrix} =
	\begin{bmatrix}
		b_1[1 - e^{-2.0t^3}] + c_1[1 - e^{-2.0t^3}]\sin(\omega_1t)  \\
		b_2[1 - e^{-1.8t^3}] + c_2[1 - e^{-1.8t^3}]\sin(\omega_2t)
	\end{bmatrix},
\end{align}
where $b_1=0.78 \ \mathrm{[rad]}$, $c_1=0.17 \ \mathrm{[rad]}$, $\omega_1=15 \ \mathrm{[rad/s]}$, $b_2=1.04 \ \mathrm{[rad]}$, $c_2=2.18 \ \mathrm{[rad]}$ and $\omega_2=3.5 \mathrm{[rad/s]}$. The explicit control law of Proposition \ref{pro3} defined in \eqref{newt0} is given by
\begin{equation}\label{eqn:controlawProp4}
\boldsymbol\tau=
\begin{bmatrix}
k_I \Sigma^7_{i=1} {Y}_{1i}\hat{\theta}_i -k_{d_1}s_1-k_{p_1}\tilde{q}_1\\
k_I \Sigma^7_{i=1} {Y}_{2i}\hat{\theta}_i-k_{d_2}s_2-k_{p_2}\tilde{q}_2\\
\end{bmatrix},
\end{equation}
where $K_p=\mathrm{diag}\{k_{p_1},k_{p_2}\}$, $K_d=\mathrm{diag}\{k_{d_1},k_{d_2}\}$, $K_s=\diag\{k_s,k_s\}$, $s_1=\dot{\tilde{q}}_1+k_s\tilde{q}_1$, and $s_2=\dot{\tilde{q}}_2+k_s\tilde{q}_2$. Moreover, $Y_{ij}$ denotes the $ij$ element of the regressor $Y(\bm{q},\dot{\bm{q}},\dot{\bm{q}}_r,\ddot{\bm{q}}_r)\in\mathbb{R}^{2\times 7}$, which is given in~\eqref{eqn:Y}, where $\dot{q}_{r_1}=\dot{q}_{1\star}-k_s\tilde{q}_1$ and $\dot{q}_{r_2}=\dot{q}_{2\star}-k_s\tilde{q}_2$.
\begin{figure*}
 \begin{equation}\label{eqn:Y}
 \begin{aligned}
Y=
\begin{bmatrix}
\ddot{q}_{r_1}&\cos(q_2)(2\ddot{q}_{r_1}+\ddot{q}_{r_2})-\sin(q_2)\big(\dot{q}_2\dot{q}_{r_1}+\dot{q}_{r_2}(\dot{q}_1+\dot{q}_2)\big)&\ddot{q}_{r_2}&g\sin(q_1)&g\sin(q_1+q_2)&\dot{q}_1^2&0\\
0&\ddot{q}_{r_1}\cos(q_2)+\sin(q_2)\dot{q}_1\dot{q}_{r_1}&\ddot{q}_{r_1}+\ddot{q}_{r_2}&0&g\sin(q_1+q_2)&0&\dot{q}_2^2
\end{bmatrix}
\end{aligned}
\end{equation}
\end{figure*}
By selecting $\Gamma=\mathrm{diag}\{\gamma_1,\ldots,\gamma_7\}$, with $\gamma_i$ strictly positive constants, the adaptation law \eqref{newt00} is explicitly given by
\begin{equation*}
\begin{split}
&\begin{bmatrix}\hat{\theta}_1 \\ \hat{\theta}_2 \\ \hat{\theta}_3 \\ \hat{\theta}_4 \\ \hat{\theta}_5 \\ \hat{\theta}_6 \\ \hat{\theta}_7\end{bmatrix}=
-k_I\int_{0}^t
\begin{bmatrix}
\gamma_1(Y_{11}(\sigma)s_1(\sigma)+Y_{21}(\sigma)s_2(\sigma))\\
\gamma_2(Y_{12}(\sigma)s_1(\sigma)+Y_{22}(\sigma)s_2(\sigma))\\
\gamma_3(Y_{13}(\sigma)s_1(\sigma)+Y_{23}(\sigma)s_2(\sigma))\\
\gamma_4(Y_{14}(\sigma)s_1(\sigma)+Y_{24}(\sigma)s_2(\sigma))\\
\gamma_5(Y_{15}(\sigma)s_1(\sigma)+Y_{25}(\sigma)s_2(\sigma))\\
\gamma_6(Y_{16}(\sigma)s_1(\sigma)+Y_{26}(\sigma)s_2(\sigma))\\
\gamma_7(Y_{17}(\sigma)s_1(\sigma)+Y_{27}(\sigma)s_2(\sigma))\\
\end{bmatrix}d\sigma\\
&+
\Delta \int_{0}^t
\begin{bmatrix}
\gamma_1(\caly_1 (\sigma)-k_I\Delta(\sigma)\hat{\theta}_1(\sigma))\\
\gamma_2(\caly_2 (\sigma)-k_I\Delta(\sigma)\hat{\theta}_2(\sigma))\\
\gamma_3(\caly_3 (\sigma)-k_I\Delta(\sigma)\hat{\theta}_3(\sigma))\\
\gamma_4(\caly_4 (\sigma)-k_I\Delta(\sigma)\hat{\theta}_4(\sigma))\\
\gamma_5(\caly_5 (\sigma)-k_I\Delta(\sigma)\hat{\theta}_5(\sigma))\\
\gamma_6(\caly_6 (\sigma)-k_I\Delta(\sigma)\hat{\theta}_6(\sigma))\\
\gamma_7(\caly_7 (\sigma)-k_I\Delta(\sigma)\hat{\theta}_7(\sigma))\\
\end{bmatrix}d\sigma+
\begin{bmatrix}\hat{\theta}_1(0) \\ \hat{\theta}_2(0) \\ \hat{\theta}_3(0) \\ \hat{\theta}_4(0) \\ \hat{\theta}_5(0)\\ \hat{\theta}_6(0) \\ \hat{\theta}_7(0) \end{bmatrix}.
\end{split}
\end{equation*}

%
\par Regarding Proposition \ref{pro4}, the explicit control law of equation \eqref{newt1} is given by:
\begin{equation}\label{eqn:controlaw4}
\boldsymbol\tau=
\begin{bmatrix}
k_I \Sigma^7_{i=1} Y_{n_{1i}} \hat \theta_i -k_{d_1}\dot{\tilde{q}}_1-k_{p_1}\tilde{q}_1\\
k_I \Sigma^7_{i=1} Y_{n_{2i}} \hat \theta_i-k_{d_2}\dot{\tilde{q}}_2-k_{p_2}\tilde{q}_2\\
\end{bmatrix},
\end{equation}
where $K_p=\mathrm{diag}\{k_{p_1},k_{p_2}\}$ and $K_d=\mathrm{diag}\{k_{d_1},k_{d_2}\}$. Moreover, $Y_{n_{ij}}$ denotes the $ij$ element of the regressor $Y_n(\bm{q},\dot{\bm{q}},\dot{\bm{q}}_\star,\ddot{\bm{q}}_\star)\in\mathbb{R}^{2\times 7}$, which is defined in~\eqref{Yn} and given explicitly in \eqref{eqn:Yn}.
\begin{figure*}
 \begin{equation}\label{eqn:Yn}
\begin{aligned}
Y_n=
\begin{bmatrix}
\ddot{q}_{1\star}&\cos(q_2)(2\ddot{q}_{1\star}+\ddot{q}_{2\star})-\sin(q_2)\big(\dot{q}_2\dot{q}_{1\star}+\dot{q}_{2\star}(\dot{q}_1+\dot{q}_2)\big)&\ddot{q}_{2\star}&g\sin(q_1)&g\sin(q_1+q_2)&\dot{q}_1^2&0\\
0&\ddot{q}_{1\star}\cos(q_2)+\sin(q_2)\dot{q}_1\dot{q}_{1\star}&\ddot{q}_{1\star}+\ddot{q}_{2\star}&0&g\sin(q_1+q_2)&0&\dot{q}_2^2
\end{bmatrix}
\end{aligned}
\end{equation}
\end{figure*}

\begin{table}[ht!]
	\caption{Least squares plus dynamic regression extension (LS+DREM) parameters.}
	\label{tab:LSDpar}
	\begin{center}
		\begin{tabular}{c|c}
          \hline
			  Parameter & Value \\
        \hline
           $\alpha$&$50$\\
           $f_0$&$30$\\
           $\beta_0$&$0.001$\\
           $\rho_0$&$20000$\\
           $\mu_0$&$0.08\mathrm{col} \{ 1,1,1,1,1,1,1 \}$ \\
            \hline
		\end{tabular}
	\end{center}
\end{table}
By selecting $\Gamma=\mathrm{diag}\{\gamma_1,\ldots,\gamma_7\}$, with $\gamma_i$ being strictly positive constants, the adaptation law \eqref{newt11} is given by:
\begin{equation*}
\begin{split}
&\begin{bmatrix}\hat{\theta}_1 \\ \hat{\theta}_2 \\ \hat{\theta}_3 \\ \hat{\theta}_4 \\ \hat{\theta}_5\\ \hat{\theta}_6\\ \hat{\theta}_7\end{bmatrix}=
-k_I\int_{0}^t
\begin{bmatrix}
\gamma_1({Y_n}_{11}(\sigma)\dot{\tilde{q}}_1(\sigma)+{Y_n}_{21}(\sigma)\dot{\tilde{q}}_2(\sigma))\\
\gamma_2({Y_n}_{12}(\sigma)\dot{\tilde{q}}_1(\sigma)+{Y_n}_{22}(\sigma)\dot{\tilde{q}}_2(\sigma))\\
\gamma_3({Y_n}_{13}(\sigma)\dot{\tilde{q}}_1(\sigma)+{Y_n}_{23}(\sigma)\dot{\tilde{q}}_2(\sigma))\\
\gamma_4({Y_n}_{14}(\sigma)\dot{\tilde{q}}_1(\sigma)+{Y_n}_{24}(\sigma)\dot{\tilde{q}}_2(\sigma))\\
\gamma_5({Y_n}_{15}(\sigma)\dot{\tilde{q}}_1(\sigma)+{Y_n}_{25}(\sigma)\dot{\tilde{q}}_2(\sigma))\\
\gamma_6({Y_n}_{16}(\sigma)\dot{\tilde{q}}_1(\sigma)+{Y_n}_{26}(\sigma)\dot{\tilde{q}}_2(\sigma))\\
\gamma_7({Y_n}_{17}(\sigma)\dot{\tilde{q}}_1(\sigma)+{Y_n}_{27}(\sigma)\dot{\tilde{q}}_2(\sigma))\\
\end{bmatrix}d\sigma\\
&+
\Delta \int_{0}^t
\begin{bmatrix}
\gamma_1(\caly_1 (\sigma)-k_I\Delta(\sigma)\hat{\theta}_1(\sigma))\\
\gamma_2(\caly_2 (\sigma)-k_I\Delta(\sigma)\hat{\theta}_2(\sigma))\\
\gamma_3(\caly_3 (\sigma)-k_I\Delta(\sigma)\hat{\theta}_3(\sigma))\\
\gamma_4(\caly_4 (\sigma)-k_I\Delta(\sigma)\hat{\theta}_4(\sigma))\\
\gamma_5(\caly_5 (\sigma)-k_I\Delta(\sigma)\hat{\theta}_5(\sigma))\\
\gamma_6(\caly_6 (\sigma)-k_I\Delta(\sigma)\hat{\theta}_6(\sigma))\\
\gamma_7(\caly_7 (\sigma)-k_I\Delta(\sigma)\hat{\theta}_7(\sigma))\\
\end{bmatrix}d\sigma+\begin{bmatrix}\hat{\theta}_1(0) \\ \hat{\theta}_2(0) \\ \hat{\theta}_3(0) \\ \hat{\theta}_4(0) \\ \hat{\theta}_5(0)\\ \hat{\theta}_6(0) \\ \hat{\theta}_7(0)\end{bmatrix}.
\end{split}
\end{equation*}
\par Regarding Proposition \ref{pro5}, the explicit control law of \eqref{newt11} is given by \eqref{eqn:controlaw4}, where $Y_n$ is given in \eqref{eqn:Yn}. By selecting $\Gamma=\mathrm{diag}\{\gamma_1,\ldots,\gamma_7\}$, with $\gamma_i$ being strictly positive constants, the adaptation law \eqref{newt31} is explicitly given by
\begin{equation*}
\begin{split}
&\begin{bmatrix}\hat{\theta}_1 \\ \hat{\theta}_2 \\ \hat{\theta}_3 \\ \hat{\theta}_4 \\ \hat{\theta}_5\\ \hat{\theta}_6\\ \hat{\theta}_7\end{bmatrix}=
-k_I\int_{0}^t
\begin{bmatrix}
\gamma_1\Big({Y_n}_{11}(\sigma) \varphi_1 +{Y_n}_{21}(\sigma)\varphi_2 \Big)\\
\gamma_2\Big({Y_n}_{12}(\sigma) \varphi_1 +{Y_n}_{22}(\sigma) \varphi_2 \Big)\\
\gamma_3\Big({Y_n}_{13}(\sigma) \varphi_1 +{Y_n}_{23}(\sigma) \varphi_2 \Big)\\
\gamma_4\Big({Y_n}_{14}(\sigma) \varphi_1 +{Y_n}_{24}(\sigma) \varphi_2 \Big)\\
\gamma_5\Big({Y_n}_{15}(\sigma) \varphi_1+{Y_n}_{25}(\sigma) \varphi_2 \Big)\\
\gamma_6\Big({Y_n}_{16}(\sigma) \varphi_1+{Y_n}_{26}(\sigma) \varphi_2 \Big)\\
\gamma_7\Big({Y_n}_{17}(\sigma) \varphi_1+{Y_n}_{27}(\sigma) \varphi_2 \Big)
\end{bmatrix}d\sigma\\
&+
\Delta \int_{0}^t
\begin{bmatrix}
\gamma_1(\caly_1 (\sigma)-k_I\Delta(\sigma)\hat{\theta}_1(\sigma))\\
\gamma_2(\caly_2 (\sigma)-k_I\Delta(\sigma)\hat{\theta}_2(\sigma))\\
\gamma_3(\caly_3 (\sigma)-k_I\Delta(\sigma)\hat{\theta}_3(\sigma))\\
\gamma_4(\caly_4 (\sigma)-k_I\Delta(\sigma)\hat{\theta}_4(\sigma))\\
\gamma_5(\caly_5 (\sigma)-k_I\Delta(\sigma)\hat{\theta}_5(\sigma))\\
\gamma_6(\caly_6 (\sigma)-k_I\Delta(\sigma)\hat{\theta}_6(\sigma))\\
\gamma_7(\caly_7 (\sigma)-k_I\Delta(\sigma)\hat{\theta}_7(\sigma))\\
\end{bmatrix}d\sigma+
\begin{bmatrix}\hat{\theta}_1(0) \\ \hat{\theta}_2(0) \\ \hat{\theta}_3(0) \\ \hat{\theta}_4(0) \\ \hat{\theta}_5(0)\\ \hat{\theta}_6(0) \\ \hat{\theta}_7(0)\end{bmatrix}.
\end{split}
\end{equation*}
with $\varphi_1(\dot {\tilde q}, \tilde q, \sigma)=\beta\dot{\tilde{q}}_1(\sigma)+ \frac{2\tilde{q}_1(\sigma)}{1+2\|\tilde{\bm{q}}(\sigma)\|^2}$ and $\varphi_2(\dot {\tilde q}, \tilde q, \sigma)=\beta\dot{\tilde{q}}_2(\sigma)+ \frac{2\tilde{q}_1(\sigma)}{1+2\|\tilde{\bm{q}}(\sigma)\|^2}$.


\par The controller gains of Propositions \ref{pro3}--\ref{pro5} were initially chosen following the tuning guidelines of Subsection \ref{sub:Gains}, while seeking to minimize the three performance indices presented in Subsection \ref{sub:performance}. Another requirement during the tuning process was to operate within a safe range of the maximum allowable torque of both actuators, which are $\pm 150 \ \mathrm{[Nm]}$ for the joint 1 actuator and $\pm 15 \ \mathrm{[Nm]}$ for the joint 2 actuator. The obtained gains are listed in Table \ref{tab:AdaptiveGains}. Note that, when applicable, the same gains were used in the three controllers to ensure a fair comparison.

\par Regarding the selection of the gain $\beta$ for the controller of Proposition \ref{pro5}, we first compute the values $\beta_1$, $\beta_2$, and $\beta_3$ from \eqref{eqn:beta}, such that $\beta>\max\{\beta_1,\beta_2,\beta_3\}$. To this end, as reported in \cite{Cicese41}--\cite{Cicese42} for the robot of Figure \ref{Fig:1to2}, $\lambda_{\max}\{M\}=5.03$, $\lambda_{\min}\{M\}=0.087$, and $k_C=0.336$. Moreover, the upper bound of the time derivative of the desired trajectory \eqref{eq:qdma} is $\| \dot{\bm{q}}_\star(t) \|=8.1$. Considering the gains $K_p$ and $K_D$ listed in Table \ref{tab:AdaptiveGains} yields $\beta_1=2.41$, $\beta_2=1.41$, and $\beta_3=2.076$. To satisfy \eqref{eqn:beta}, we set $\beta=3$.

\begin{table}[h]
	\caption{Parameters of the controllers in Propositions~\ref{pro3}--\ref{pro5}, as well as of the PD$+$AC controller in~\eqref{eqn:tauslotineli1}--\eqref{eqn:tauslotineli2}.}
	\label{tab:AdaptiveGains}
\begin{center}
 \begin{threeparttable}
		\begin{tabular}{ccc} \hline
			            & Value & Units \\
        \hline
			$k_{p_1}$  & $500$ & $\mathrm{Nm/rad}$ \\
			$k_{p_2}$  & $200$ & $\mathrm{Nm/rad}$ \\
			$k_{d_1}$  & $10$ & $\mathrm{Nm \ s/rad}$ \\
			$k_{d_2}$  & $10$ & $\mathrm{Nm \ s/rad}$ \\
            $k_I$      & $0.75$ & $-$ \\
            $k_{s_1}^{\ast\ast}$      & $3$ & $-$\\
            $k_{s_2}^{\ast\ast}$      & $3$ & $-$\\
			$\gamma_1$  & $0.1$ & $\mathrm{kg \ s^2/m^2}$  \\
            $\gamma_2$  & $0.025$ & $\mathrm{kg \ s^2/m^2}$  \\
			$\gamma_3$  & $0.1$ & $\mathrm{kg \ s^2/m^2}$  \\
            $\gamma_4$  & $0.5$ & $\mathrm{kg \ s^2/m^2}$  \\
            $\gamma_5$  & $0.1$ & $\mathrm{kg \ s^2/m^2}$  \\
            $\gamma_6$  & $0.5$ & $\mathrm{kg \ s^2/m^2}$  \\
            $\gamma_7$  & $0.075$ & $\mathrm{kg \ s^2/m^2}$  \\
            $\hat{\boldsymbol\theta}(0)$ & $\bm{0}_7$ & $\ast$ \\
            \hline
		\end{tabular}

\begin{tablenotes}
            \item $\ast$ \textit{As indicated in Table \ref{tab:parunits}}.
            \item $\ast\ast$ \textit{Only for controller of Proposition 3}.
        \end{tablenotes}
    \end{threeparttable}
    \end{center}
\end{table}
\begin{table}[h]
	\caption{Parameters of the CLAC of \eqref{eqn:leycontrol}-\eqref{eqn:leyadaptacion}.}
	\label{tab:AdaptiveGains2}
\begin{center}
 \begin{threeparttable}
	
		\begin{tabular}{ccc} \hline
			            & Value & Units \\
        \hline
            $\varsigma$ & $1$ & $-$\\
			$\varrho_{11}$  & $13.3$ & $-$ \\
			$\varrho_{22}$  & $50$ & $-$ \\
			$\kappa$  & $0.000001$ & $-$ \\
             $k_{c_1}$      & $150$ & $\mathrm{Nm \ s/rad}$ \\
            $k_{c_2}$      & $15$ & $\mathrm{Nm \ s/rad}$ \\
			$\gamma_1$  & $0.05$ & $\mathrm{kg \ s^2/m^2}$  \\
            $\gamma_2$  & $0.01$ & $\mathrm{kg \ s^2/m^2}$  \\
			$\gamma_3$  & $0.01$ & $\mathrm{kg \ s^2/m^2}$  \\
            $\gamma_4$  & $0.5$ & $\mathrm{kg \ s^2/m^2}$  \\
            $\gamma_5$  & $0.05$ & $\mathrm{kg \ s^2/m^2}$  \\
            $\gamma_6$  & $0.5$ & $\mathrm{kg \ s^2/m^2}$  \\
            $\gamma_7$  & $0.05$ & $\mathrm{kg \ s^2/m^2}$  \\
            $\hat{\boldsymbol\theta}(0)$ & $\bm{0}_7$ & $\ast$ \\
            $\tau_d$ & $2$ & $\mathrm{s}$\\
            $c_w$ & $10$ & $-$ \\
            \hline
		\end{tabular}
\begin{tablenotes}
            \item $\ast$ \textit{As indicated in Table \ref{tab:parunits}}.
        \end{tablenotes}
    \end{threeparttable}
    \end{center}
\end{table}
\subsection{A comparative study}
In this section, we present two controllers from the literature that are used to perform a comparative study against our proposals, namely, the composite learning-based trajectory-tracking controller of \cite{yongpinpan} and the classical direct adaptive PD$+$adaptive-compensation controller of \cite{SLOLItac}. As stated in the introduction, the adaptive trajectory-tracking composite learning-based controller of \cite{yongpinpan} ensures exponential convergence to zero of the joint position error, the velocity error, and the estimation error, that is, the difference between the estimated unknown parameters and their true values, without imposing the PE condition. Instead, it relies on the relaxed IE condition. Therefore, it ensures essentially the same stability conclusions as the controllers of Propositions \ref{pro3}--\ref{pro5}, making it an ideal candidate for comparison with an up-to-date controller from the literature.


\par On the other hand, the classical \emph{direct} adaptive trajectory-tracking controller of \cite{SLOLItac} ensures convergence to zero of the joint position and velocity errors. However, convergence to zero of the parametric error is guaranteed only if the system dynamics regressor satisfies the PE condition. Thus, to illustrate the importance of the excitation condition, we experimentally evaluate the controller of \cite{SLOLItac} and compare its results with those obtained using the controllers of Propositions \ref{pro3}--\ref{pro5}. Particular emphasis is placed on the comparison with the controller of Proposition \ref{pro3}, since both controllers have almost the same structure. The results show that the proposed improvement not only enhances the tracking performance, but also improves the convergence of the parameter estimates.

\subsubsection{A comparison with a composite learning-based adaptive controller}
The \textit{composite learning adaptive controller} (CLAC) presented in \cite{yongpinpan} is given by the control law:
\begin{eqnarray}
  \label{eqn:leycontrol}
  \boldsymbol\tau&=&K_c\bm{e}_f+\Phi^T\hat{\boldsymbol\theta},\\
   \label{eqn:leyadaptacion}
  \dot{\hat{\boldsymbol\theta}}&=&\Gamma\mathcal{P}\big(\Phi^T\bm{e}_f+\kappa\boldsymbol\epsilon\big),
\end{eqnarray}
with $\bm{e}_f:=\dot{\tilde{\bm{q}}}+\Lambda\tilde{\bm{q}}$, where $\Lambda\in\mathbb{R}^{n_{q}\times n_{q}}$ is a diagonal positive definite matrix, $\tilde{\bm{q}}(t)=\bm{q}_\star(t)-\bm{q}(t)$, and $\bm{q}_\star(t)$ is the desired tracking trajectory. Moreover, $\Phi(\bm{q},\dot{\bm{q}},\bm{v},\dot{\bm{v}})\in\mathbb{R}^{n_q \times w}$ is obtained by invoking Property \ref{prop1}, with $\Phi=Y$, where $\bm{v}=\dot{\bm{q}}_\star+\Lambda\tilde{\bm{q}}$. The matrix $K_c\in\mathbb{R}^{n_{q}\times n_{q}}$ is constant, diagonal, and positive definite; $\kappa\in\mathbb{R}$ is a strictly positive constant; and $\Gamma:=\mathrm{diag}\{\gamma_{1},\ldots,\gamma_{w}\}$, with $\gamma_{i}\in\mathbb{R}$ strictly positive constants for $i\in\{1,\ldots,w\}$. Finally, $\mathcal{P}\big(\Phi^T\bm{e}_f+\kappa\boldsymbol\epsilon\big)$ is the projection operator defined as:
\begin{equation}\label{eqn:proyeccion}
\begin{split}
&\mathcal{P}\big(\Phi^T\bm{e}_f+\kappa\boldsymbol\epsilon\big):=\\
  &\begin{dcases}
\Phi^T\bm{e}_f+\kappa\boldsymbol\epsilon &  \text{if }\| \hat{\boldsymbol\theta}\|<c_w , \\
\Phi^T\bm{e}_f+\kappa\boldsymbol\epsilon & \begin{multlined} \text{if }\|\hat{\boldsymbol\theta}\|=c_w \text{ and } \\ \hat{\boldsymbol\theta}^T\big[\Phi^T\bm{e}_f+\kappa\boldsymbol\epsilon \big]\leq 0, \end{multlined}\\
\!\begin{multlined}
\big[\Phi^T\bm{e}_f+\kappa\boldsymbol\epsilon \big] \\ -\frac{\hat{\boldsymbol\theta}\hat{\boldsymbol\theta}^T\big[\Phi^T\bm{e}_f+\kappa\boldsymbol\epsilon \big]}{\|\hat{\boldsymbol\theta}\|^2}
\end{multlined} &
  \text{if }\| \hat{\boldsymbol\theta}\|>c_w, \\
\!\begin{multlined}
\big[\Phi^T\bm{e}_f+\kappa\boldsymbol\epsilon \big] \\ -\frac{\hat{\boldsymbol\theta}\hat{\boldsymbol\theta}^T\big[\Phi^T\bm{e}_f+\kappa\boldsymbol\epsilon \big]}{\|\hat{\boldsymbol\theta}\|^2}
\end{multlined}& \begin{multlined}\text{if }\|\hat{\boldsymbol\theta}\|=c_w \text{ and } \\ \hat{\boldsymbol\theta}^T\big[\Phi^T\bm{e}_f+\kappa\boldsymbol\epsilon \big]> 0,\end{multlined}
    \end{dcases}
    \end{split}
\end{equation}
with $c_w\in\mathbb{R}^{+}$ a known constant and $\boldsymbol\epsilon\in\mathbb{R}^{w}$ is the prediction error given by:
  \begin{equation}\label{eqn:eps}
\boldsymbol\epsilon(t)=
  \begin{dcases}
\boldsymbol{y}(t_e)-\Theta(t_e)\hat{\boldsymbol\theta}(t)&   \text{for } t\geq T_e,\\
\boldsymbol{0}_w& \text{for } 0\leq t<T_e,
    \end{dcases}
\end{equation}
where:
\begin{equation}\label{eqn:yCL}
\boldsymbol{y}:=\int_{t-\tau_d}^t\Phi_f(\tau)\boldsymbol\tau_f(\tau) d_{\tau}\in\mathbb{R}^{w},
\end{equation}
\begin{equation}\label{eqn:theta}
\Theta:=\int_{t-\tau_d}^t\Phi_f(\tau)\Phi^T_f(\tau) d_{\tau}\in\mathbb{R}^{w\times w},
\end{equation}
with $\Phi_f(\bm{q},\dot{\bm{q}}):=\varsigma e^{-\varsigma t}\ast\chi(\bm{q},\dot{\bm{q}},\dot{\bm{q}},\ddot{\bm{q}})$, where $\varsigma\in\mathbb{R}^{+}$ is a design constant, and $\chi(\bm{q},\dot{\bm{q}},\dot{\bm{q}},\ddot{\bm{q}})$ is obtained by invoking Property \ref{prop1}, with $\chi=Y$. Here, ``$\ast$'' denotes the convolution operator, such that $\boldsymbol\tau_f:=\varsigma e^{-\varsigma t}\ast \boldsymbol\tau$, and $T_e$ denotes the epoch at which $\Theta$ in \eqref{eqn:theta} satisfies the IE condition, that is, the time instant at which $\Theta(T_e)\geq\sigma I_{w\times w}$. The constant $\sigma\in\mathbb{R}$ is determined by the minimum singular value of the matrix $\boldsymbol\Theta$ in \eqref{eqn:theta}, and $t_e(t):=\mathrm{arg \ max}\{ \sigma(\tau)\}_{\tau\in[T_e,t]}$.
\par For the robot of Figure \ref{Fig:1to2} considering \eqref{thesthe}, the explicit control law of the CLAC of \eqref{eqn:leycontrol} is given by
\begin{equation}
\boldsymbol\tau=
\begin{bmatrix}
k_{c_1}e_{f_1}+\Sigma_{i=1}^{7}\phi_{1i}\hat{\theta}_i\\
k_{c_2}e_{f_2}+\Sigma_{i=1}^{7}\phi_{2i}\hat{\theta}_i
\end{bmatrix},
\end{equation}
where $e_{f_1}=\dot{\tilde{q}}_1+\varrho_{11}\tilde{q}_1$ and $e_{f_2}=\dot{\tilde{q}}_2+\varrho_{22}\tilde{q}_2$, with $\Lambda=\mathrm{diag}\{\varrho_{11},\varrho_{22}\}$, where $\varrho_{11},\varrho_{22}\in\mathbb{R}$ are strictly positive constants. Moreover, $K_c=\mathrm{diag}\{k_{c_1},k_{c_2}\}$, and the matrix $\Phi\in\mathbb{R}^{7\times 2}$ is given by \eqref{eqn:Phi}, where $v_1=\dot{q}_{1_\star}+\lambda_{11}\tilde{q}_1$ and $v_2=\dot{q}_{2_\star}+\lambda_{22}\tilde{q}_2$. The adaptation law in \eqref{eqn:leyadaptacion} can be computed by evaluating the projection operator \eqref{eqn:proyeccion}. To this end, the matrix $\Phi_f$, explicitly given in \eqref{eqn:Phif}, is required to evaluate the signal $\epsilon$ in \eqref{eqn:eps}.


\begin{figure*}[h!]
 \begin{equation}\label{eqn:Phi}
\begin{aligned}
\Phi^T=
\begin{bmatrix}
v_1&2v_1\cos(q_2)+v_2\cos(q_2)-\dot{v}_1\sin(q_2)\dot{q}_2-\dot{v}_2\sin(q_2)(\dot{q}_1+\dot{q}_2)&v_2&g\sin(q_1)&g\sin(q_1+q_2)&\dot{q}_1^2&0\\
0&v_1\cos(q_2)+\dot{v}_1\sin(q_2)\dot{q}_1&v_1+v_2&0&g\sin(q_1+q_2)&0&\dot{q}_2^2
\end{bmatrix}
\end{aligned}
\end{equation}
\end{figure*}

 \begin{figure*}[h!]
 \begin{equation}\label{eqn:Phif}
\begin{aligned}
\Phi_f^T=
\begin{bmatrix}
\frac{\varsigma p}{p+\varsigma}[\dot{q}_1] &  \frac{\varsigma p}{p+\varsigma}\big[(2\dot{q}_1+\dot{q}_2)\cos(q_2) \big] & \frac{\varsigma p}{p+\varsigma}[\dot{q}_2] & \frac{\varsigma }{p+\varsigma}[g\sin(q_1)] & \frac{\varsigma }{p+\varsigma}[g\sin(q_1+q_2)] & \frac{\varsigma }{p+\varsigma}[\dot{q}_1^2] & 0 \\
0 &  \frac{\varsigma p}{p+\varsigma}[\dot{q}_1\cos(q_2)] +\frac{\varsigma }{p+\varsigma}[\dot{q}_1^2\sin(q_2)] & \frac{\varsigma p}{p+\varsigma}[\dot{q}_1+\dot{q_2}] & 0 & \frac{\varsigma }{p+\varsigma}[g\sin(q_1+q_2)] & 0 & \frac{\varsigma }{p+\varsigma}[\dot{q}_2^2]
\end{bmatrix}
\end{aligned}
\end{equation}

\end{figure*}
Table \ref{tab:AdaptiveGains2} lists the CLAC controller gains from \eqref{eqn:leycontrol} and \eqref{eqn:leyadaptacion}, which were selected to minimize the three performance indices presented in Subsection \ref{sub:performance} considering also being within a safe range of the maximum allowable torque of both actuators. Thus, the same tuning criteria used for the controllers of Propositions \ref{pro3}--\ref{pro5} were followed.
 \begin{figure}[!h]
\centering
\subfloat[]{%
\resizebox*{7.5cm}{!}{\includegraphics{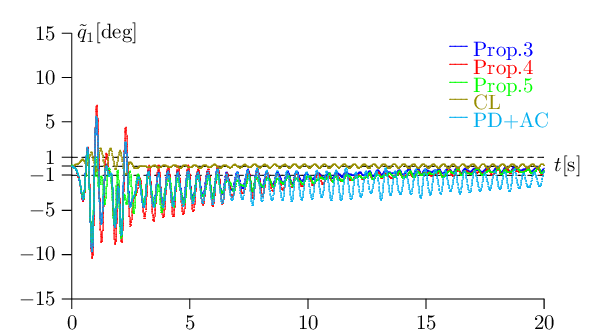}}}\\
\subfloat[]{%
\resizebox*{7.5cm}{!}{\includegraphics{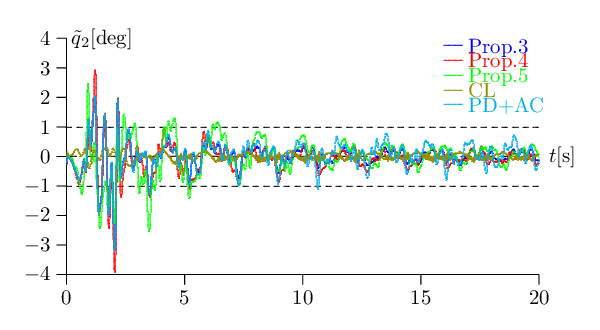}}}
\caption{Time evolution of the joint position tracking errors $\tilde{q}_1(t)$ and $\tilde{q}_2(t)$.} \label{Fig:2}
\end{figure}
\begin{figure}[!h]
\centering
\subfloat[]{%
\resizebox*{7.5cm}{!}{\includegraphics{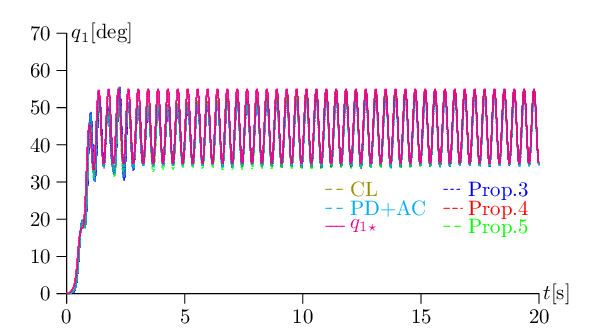}}}\\
\subfloat[]{%
\resizebox*{7.5cm}{!}{\includegraphics{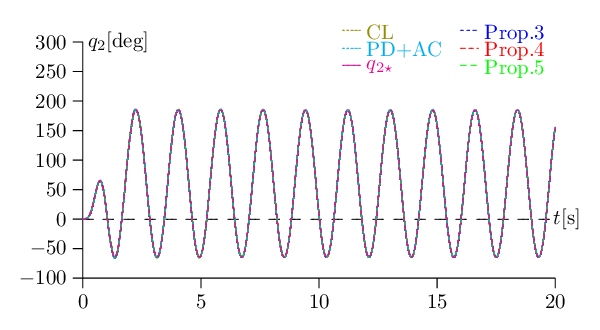}}}
\caption{Time evolution of the joint position ${q}_1(t)$ and ${q}_2(t)$ and the desired joint trajectories given by \eqref{eq:qdma}.} \label{Fig:3}
\end{figure}
\subsubsection{A comparison with the classical PD control plus adaptive compensation}
The \textit{direct adaptive version} of the PD control plus adaptive compensation originally presented in \cite{SLOLItac} is given by:
\begin{eqnarray}
\label{eqn:tauslotineli1}
\boldsymbol\tau&=&k_IY(\bm{q},\dot{\bm{q}},\dot{\bm{q}}_r, \ddot{\bm{q}}_r)\hat{\theta}-K_D\bm{s}-K_p\tilde{\bm{q}},\\
\label{eqn:tauslotineli2}
\hat{\boldsymbol\theta}(t)&=&-k_I\Gamma\int_{0}^tY(\sigma)^Ts(\sigma)d\sigma
\end{eqnarray}
where $Y(\bm{q},\dot{\bm{q}},\dot{\bm{q}}_r, \ddot{\bm{q}}_r)$ is obtained from \eqref{eqn:phi2}, $\bm{s}$ and $\dot{\bm{q}}_r$ are defined in \eqref{SLerrors}, $K_p,K_D\in\mathbb{R}^{{n_q}\times {n_q}}$ are diagonal positive definite matrices, $k_I>0$, and $\Gamma\in\mathbb{R}^{w\times w}$ is positive definite matrix. As stated in \cite{SLOLItac}, the resultant closed-loop system obtained from control law \eqref{eqn:tauslotineli1}-\eqref{eqn:tauslotineli2} in \eqref{elsys} verifies that:
\begin{equation}
\lim_{t\rightarrow \infty}\begin{bmatrix}\tilde{\bm{q}}(t)\\ \dot{\tilde{\bm{q}}}(t) \end{bmatrix}=\bm{0}_{2n},
\end{equation}
while if $Y(\cdot)$ does not satisfy the PE condition, then $\tilde{\boldsymbol\theta}$ remains bounded, that is, $\tilde{\boldsymbol\theta}\in\mathcal{L}_{\infty}$.
\par Table~\ref{tab:AdaptiveGains} lists the gains of the PD plus adaptive-compensation controller in~\eqref{eqn:tauslotineli1}--\eqref{eqn:tauslotineli2}. These gains were selected to minimize the three performance indices presented in Subsection \ref{sub:performance}, while also operating within a safe range of the maximum allowable torque of both actuators. Thus, the same tuning criteria used for the controllers of Propositions \ref{pro3}--\ref{pro5} and the CLAC controller \eqref{eqn:leycontrol}--\eqref{eqn:leyadaptacion} were followed.

 Note that, when possible, the same gain values were used for the controllers in Propositions~\ref{pro3}--\ref{pro5} and for the PD plus adaptive-compensation controller.
 \par The results of all experiments are presented in Figures \ref{Fig:2}-\ref{Fig:8}. All transient signals related to the experiment using the controller of Proposition \ref{pro3} are indicated in $\textcolor{blue}{blue}$ and labeled as $Prop.3$, all transient signals related to the experiment the controller of Proposition \ref{pro4} are indicated in $\textcolor{red}{red}$ and labeled as $Prop.4$, all transient signals related to the experiment the controller of Proposition \ref{pro5} are indicated in  $\textcolor{green}{green}$ and labeled as $Prop.5$, and all transient signals related to the experiment the CLAC of \eqref{eqn:leycontrol}-\eqref{eqn:leyadaptacion} are indicated in  $\textcolor{olive}{olive}$ and labeled as $CL$ and all transient signals related to the experiment the \emph{classical} PD control of \eqref{eqn:tauslotineli1}-\eqref{eqn:tauslotineli2} are indicated in  $\textcolor{cyan}{cyan}$ and labeled as $PD+AC$. As seen from the Figures \ref{Fig:2}-\ref{Fig:3}, the joint position trajectory tracking control objective is verified in all cases, that is, all controllers ensured that $\bm{q}\rightarrow \bm{q}_\star$. The CLAC exhibits a smaller maximum bound on the joint position trajectory tracking error for both joints, as shown in Figure \ref{Fig:2}. However, the controllers presented in Propositions \ref{pro3}–\ref{pro5} achieve faster convergence of the parametric error, as shown in Figures \ref{Fig:5}-\ref{Fig:6b}, even though the control torque is smaller in amplitude than that required by the CLAC, as seen in Figure \ref{Fig:4}. The PD+AC parametric error signals confirms that when no persistent excitation condition is fulfilled the parametric error only remains bounded. Regarding Assumption \ref{ass1} on the regressor $\boldsymbol{\Omega}$, Figure \ref{Fig:7} shows that the inequality in \eqref{del} is satisfied when using the controllers from Propositions \ref{pro3}–\ref{pro5}. Figure \ref{Fig:7} also depicts the behavior of the excitation level signal $\sigma$ for the CLAC, confirming that the IE condition for the regressor $\Theta$ in \eqref{eqn:theta} is fulfilled approximately one second after the experiment begins ($\sigma > 0$). It is interesting to note in Figure \ref{Fig:7} that the maximum value of $\sigma$ is reached at approximately $2.8\mathrm{[s]}$, resulting in an estimated $\sigma_{\max}$ of around $0.1$. This implies, in other words, that parametric error convergence is achieved using $\Theta(2.8)$ and $\boldsymbol{y}(2.8)$. It is important to emphasize that, although the joint position trajectory tracking objective is successfully achieved in all cases, none of the estimated parameters $\hat{\theta}_i$ converges exactly to its true value. This scenario can be attributed to the intricacy of the friction phenomenon, alongside implementation challenges, including quantization within DACs, finite resolution in positional measurement via incremental encoders, discrete approximation of continuous-time dynamic controllers due to ODE solving techniques or integration numeric engines, imperfect feedback on joint velocity due to estimation via coarse Euler methods from joint positional readings among other issues.

\begin{figure}[!htp]
\centering
\subfloat[Using Proposition \ref{pro3}.]{%
\resizebox*{7.0cm}{!}{\includegraphics{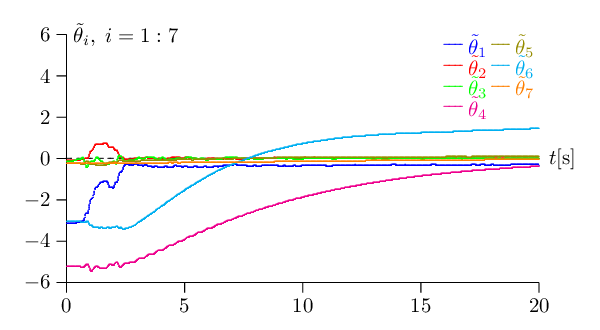}}}\\
\subfloat[Using Proposition \ref{pro4}.]{%
\resizebox*{7.0cm}{!}{\includegraphics{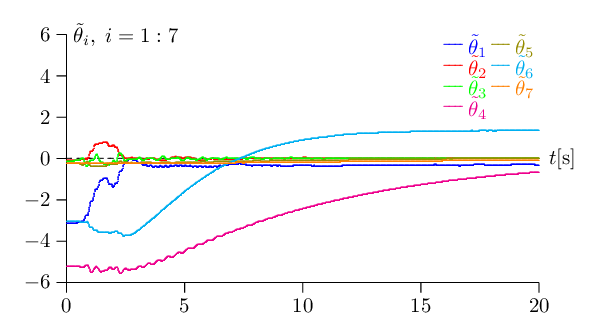}}}\\
\subfloat[Using Proposition \ref{pro5}.]{%
\resizebox*{7.0cm}{!}{\includegraphics{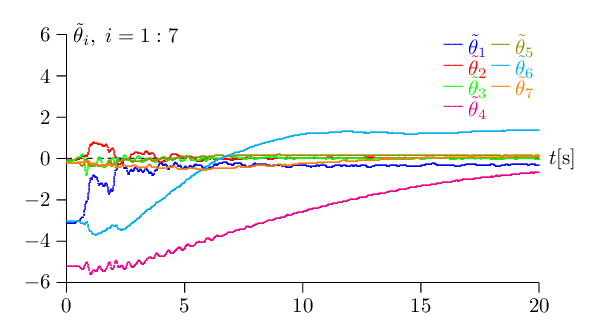}}}\\
\subfloat[Using CLAC.]{%
\resizebox*{7.0cm}{!}{\includegraphics{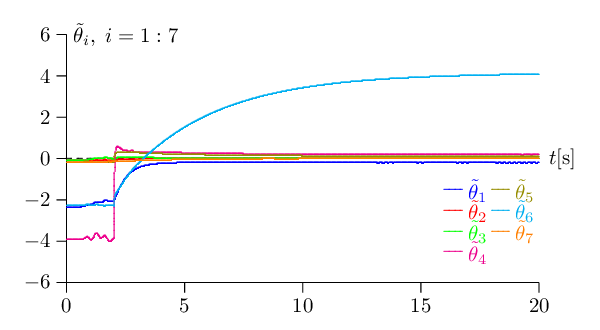}}}\\
\subfloat[Using PD+AC.]{%
\resizebox*{7.0cm}{!}{\includegraphics{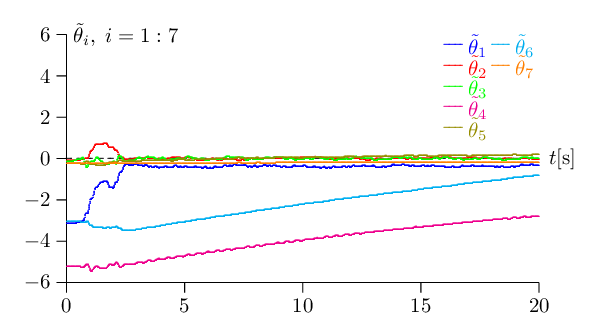}}}
\caption{Transient behavior of the estimation errors.} \label{Fig:5}
\end{figure}

\begin{figure}[!htp]
\centering
\subfloat[]{%
\resizebox*{7.0cm}{!}{\includegraphics[width=7.0cm]{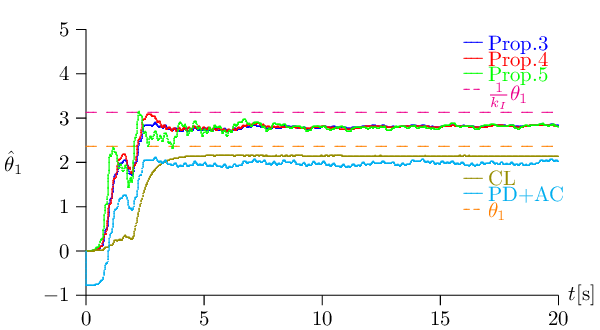}}}\\
\subfloat[]{%
\resizebox*{7.0cm}{!}{\includegraphics[width=7.0cm]{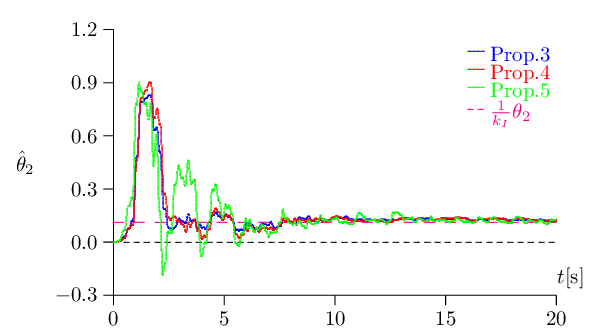}}}\\
\subfloat[]{%
\resizebox*{7.0cm}{!}{\includegraphics[width=7.0cm]{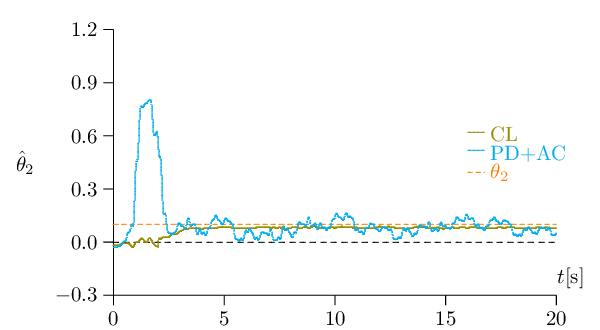}}}\\
\subfloat[]{%
\resizebox*{7.0cm}{!}{\includegraphics[width=7.0cm]{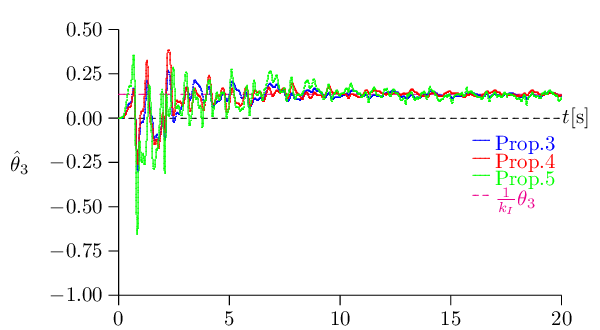}}}\\
\subfloat[]{%
\resizebox*{7.0cm}{!}{\includegraphics[width=7.0cm]{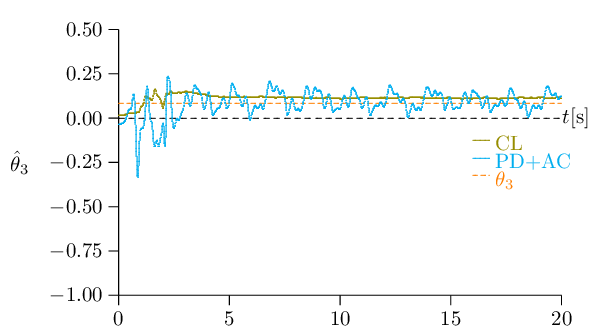}}}\\
\caption{Transient behavior of the estimated parameters.} \label{Fig:6}
\end{figure}

\begin{figure}[htp!]
\centering
\subfloat[]{%
\resizebox*{7.5cm}{!}{\includegraphics[width=7.5cm]{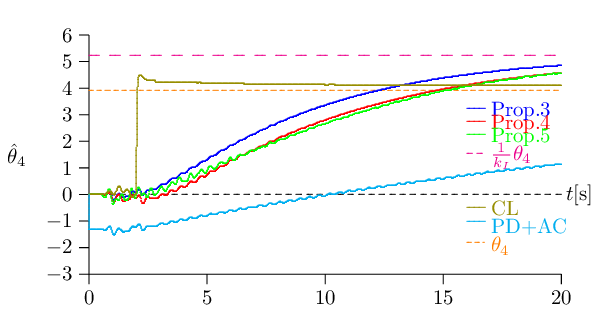}}}\\
\subfloat[]{%
\resizebox*{7.5cm}{!}{\includegraphics[width=7.5cm]{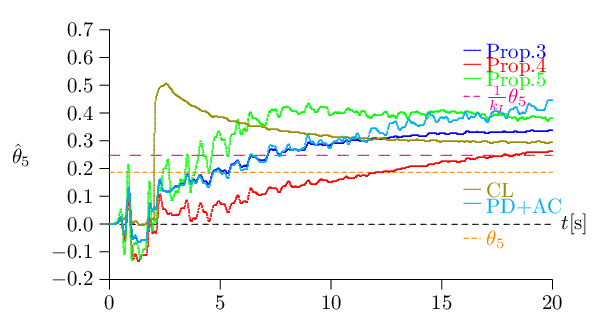}}}\\
\subfloat[]{%
\resizebox*{7.5cm}{!}{\includegraphics[width=7.5cm]{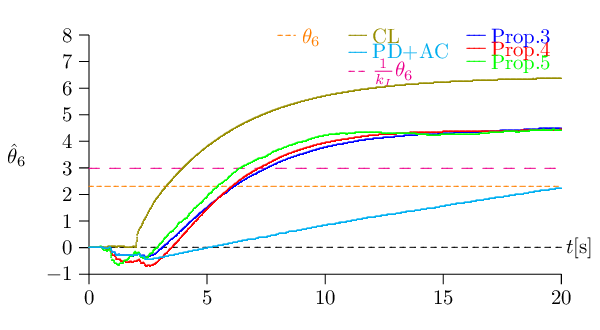}}}\\
\subfloat[]{%
\resizebox*{7.5cm}{!}{\includegraphics[width=7.5cm]{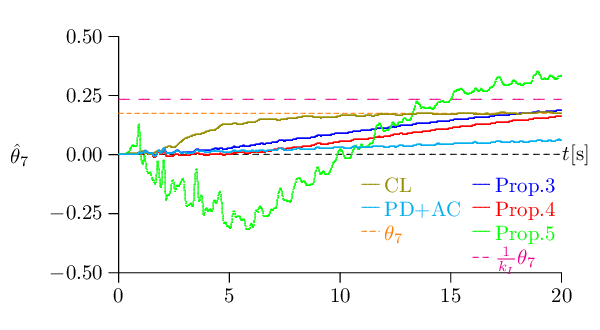}}}\\
\caption{Transient behavior of the estimated parameters.} \label{Fig:6b}
\end{figure}

\begin{figure}[htp!]
\centering
\subfloat[]{%
\resizebox*{7.5cm}{!}{\includegraphics{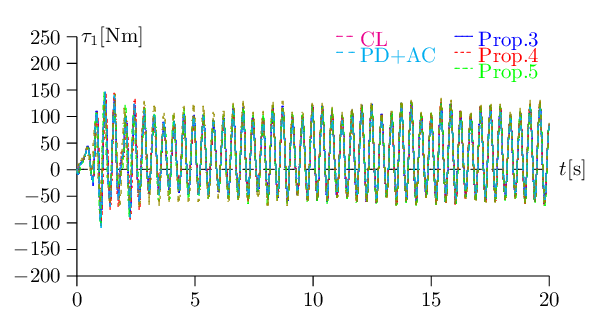}}}\\
\subfloat[]{%
\resizebox*{7.5cm}{!}{\includegraphics{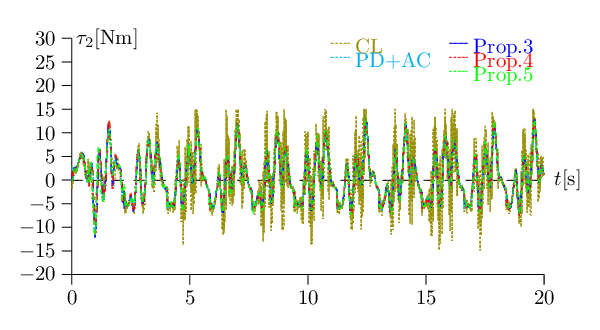}}}
\caption{Time evolution of the control inputs $\tau_1(t)$ and $\tau_2(t)$.} \label{Fig:4}
\end{figure}

\begin{figure}[htp!]
\centering
\subfloat[Transient behavior of the LS+D signal $\Delta$ of \eqref{del}.]{%
\resizebox*{7.5cm}{!}{\includegraphics{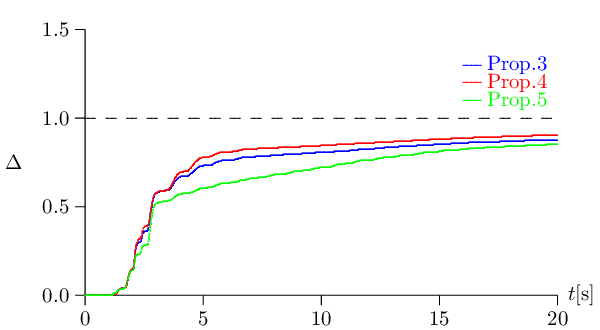}}}\\
\subfloat[Transient behavior of the CLAC signal $\sigma$ of \eqref{eqn:leyadaptacion}-\eqref{eqn:leycontrol}.]{%
\resizebox*{7.5cm}{!}{\includegraphics{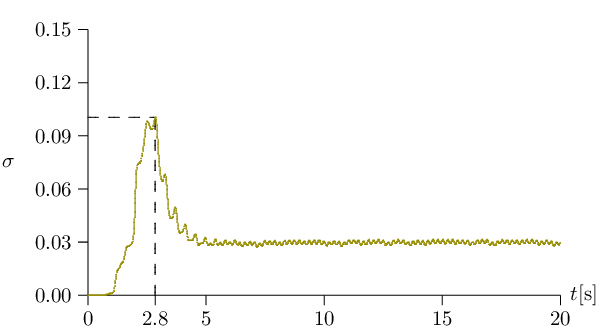}}}
\caption{Verification of the IE condition.} \label{Fig:7}
\end{figure}

 \subsection{Discussion about the computational cost}\label{sub:computational}
 The computational cost of an algorithm in practice can be evaluated in several ways. 
 
 A common approach is to measure the wall--clock execution time on a specific computer; however, this measure is strongly hardware--dependent and therefore does not provide an intrinsic indication of algorithmic efficiency. Another approach is to count the floating--point operations (flops) required, although this measure alone may be insufficient when different execution paths occur depending on the data (e.g., due to pivoting or conditional branches). A more general and hardware--independent measure is given by the computational complexity of the algorithm, which indicates how the required computational effort scales as the problem dimension increases. The \textit{Big-O} notation provides a formal way to express this scaling behavior. Let the input size of the algorithm be represented by $n$. If the number of elementary operations grows proportionally to $n^{2}$, the algorithm is said to have quadratic complexity, $O(n^{2})$. Similarly, if the required operations grow proportionally to $n^{3}$, the algorithm has cubic complexity, $O(n^{3})$. Typical relationships among complexities are
\begin{equation*}
O(\log n) < O(n) < O(n \log n) < O(n^{2}) < O(n^{3}) ,
\end{equation*}
showing, for instance, that a cubic--order algorithm requires more computational steps than a quadratic one. In the composite-learning controller of \cite{yongpinpan}, the IE condition must be verified by the matrix $\Theta$ in \eqref{eqn:theta}. This requirement involves the existence of positive constants $\tau_d$ and $\sigma$ such that
$\int_{t-\tau_d}^{t}{\phi}_f(\tau){\phi}_f^{\top}(\tau)d\tau \geq \sigma I_{w\times w}$,
where $\sigma$ is obtained by evaluating the minimum singular value of $\Theta$ in \eqref{eqn:theta}. In practice, this condition must be monitored online (as the composite learning control law \eqref{eqn:leycontrol}-\eqref{eqn:leyadaptacion} depends on the $\epsilon$ signal \eqref{eqn:eps} which is defined based on the accomplish of the IE). In contrast, in the proposed formulation (controllers of Propositions~\ref{pro3}--\ref{pro5}), the IE condition depends on the scalar regressor $\Delta$ defined in \eqref{eqn:Delta}; hence, its implementation avoids the online computation of a minimum singular value. To provide a concrete indication of the arithmetic effort required by these online computations, we compare the dominant floating--point operation counts of representative numerical procedures. 

According to \cite{golub2013matrix}, for a class of algorithms known as the R--SVD procedure, computing only the singular values of an $n\times n$ matrix requires approximately $4n^{3}$ floating--point operations. On the other hand, the proposed algorithm requires computing $\Delta(t)$ in \eqref{eqn:Delta} and the signal $\caly$ in \eqref{eqn:Caly}. Let $\Psi:=I_w- z f_{0}F$. Rather than forming $\adj\{\Psi\}$ via cofactors, we compute the required quantities using the identity $\adj\{\Psi\}=\det\{\Psi\}\Psi^{-1}=\Delta\Psi^{-1}$, implemented via one LU factorization plus two triangular substitutions \cite{horn2012matrix}. The LU factorization costs $\tfrac{2}{3}n^{3}$ flops and the two substitutions cost $O(n^{2})$ \cite{golub2013matrix}. Hence, the dominant arithmetic cost of the proposed online computation is approximately $\tfrac{2}{3}n^{3}+O(n^{2})$. For $n=7,$ this yields approximately $\tfrac{2}{3}\cdot 7^{3}+2\cdot 7^{2}\approx 229 + 98 \approx 327$ flops, up to lower-order terms and implementation details (e.g., pivoting). Consequently, both the minimum-singular-value computation and the LU-based computation belong to the same asymptotic class $O(n^{3})$, but the LU-based computation has a strictly smaller constant factor. In the specific case $n=7$ the arithmetic effort is reduced from roughly $4\cdot 7^{3}=1372$ flops to roughly $327$ flops, i.e., by a factor of approximately $\frac{1372}{327}\approx 4.2.$ Therefore, replacing the online minimum singular value evaluation of $\Theta$ with the LU-based computation required for the proposed LS$+$DREM adaptation algorithm \eqref{LSD1}--\eqref{LSD3} reduces the computational burden while preserving the proposed excitation/regressor mechanism.
 \par Finally, it is important to highlight the memory cost associated with implementing the CLAC, particularly when computing the matrix~$\Theta$ in~\eqref{eqn:theta}. Since each element of the $7\times 7$ matrix must be stored at every sampling instant over the integration window~$\tau_d$, the total memory requirement becomes significant. In the experimental setup, the sampling period was $0.0025\,\mathrm{s}$ and the parameter $\tau_d$ was set to $2\,\mathrm{s}$. Therefore, the number of samples required to compute $\Theta$ is given by:
\begin{equation*}
\frac{\tau_d}{\text{sampling period}} = \frac{2\,\mathrm{s}}{0.0025\,\mathrm{s}} = 800.
\end{equation*}
As a result, each of the $49$ elements of $\Theta$ must be stored for $800$ time steps, leading to a total memory requirement of $49 \times 800 = 39{,}200$ entries. Similarly, a comparable amount of memory is required to store the vector~$\boldsymbol{y}$ of \eqref{eqn:yCL}, as it must also be evaluated over the same integration window. This contrasts with the memory requirements of the proposed controllers of Propositions~\ref{pro3}-\ref{pro5}, for which only instantaneous data is needed to compute all integral states.

\subsection{Performance evaluation}\label{sub:performance}
The control performance is quantified using three different performance indices, which were employed as tuning criteria for all the controllers presented, as described above. First, we compute the root-mean-square value of the joint position tracking error, given by
\begin{equation}
\label{eq:norm}
e_{{\mathrm{RMS}}}\bigl[\tilde{\bm{q}}(t)\bigl] = \sqrt{\frac{1}{T}\int_{0}^{T}\tilde{\bm{q}}(t)^T\tilde{\bm{q}}(t)dt},
\end{equation}
where $T = 20$ [s] represents the total experimentation time. The smaller $e_{{\rm RMS}}$ means the  better performance, which was obtained for the CLAC of \eqref{eqn:leycontrol}-\eqref{eqn:leyadaptacion} over the proposed controllers of Propositions \ref{pro3}-\ref{pro5} and the PD+AC (see Figure \ref{Fig:9}a).
\begin{figure}[htp!]
\centering
\subfloat[]{%
\resizebox*{7.5cm}{!}{\includegraphics{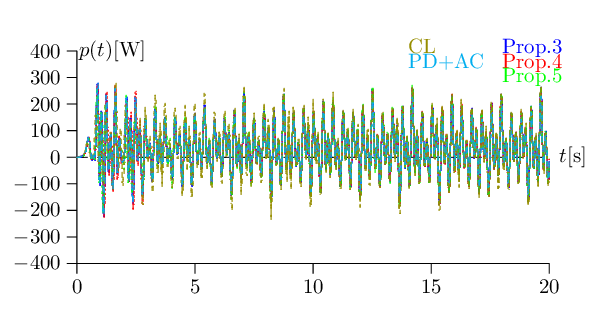}}}
\caption{Instantaneous power supplied} \label{Fig:8}
\end{figure}
\begin{figure}[h!]
\centering
\subfloat[Root mean square value of joint positions errors $\tilde{\bm{q}}$ defined by $e_{{\rm RMS}}$ in \eqref{eq:norm}.]{%
\resizebox*{7.5cm}{!}{\includegraphics{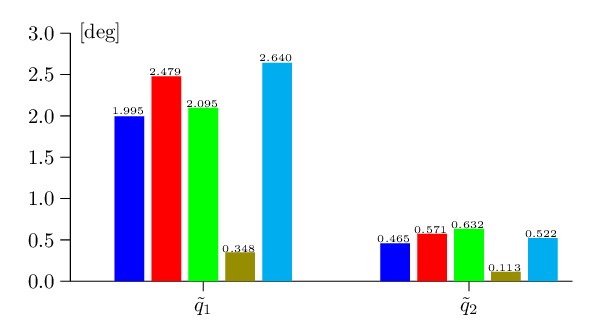}}} \\
\subfloat[Average supplied power performance index defined by $P_{\rm{avg}}$ in \eqref{eq:PE}.]{%
\resizebox*{7.5cm}{!}{\includegraphics{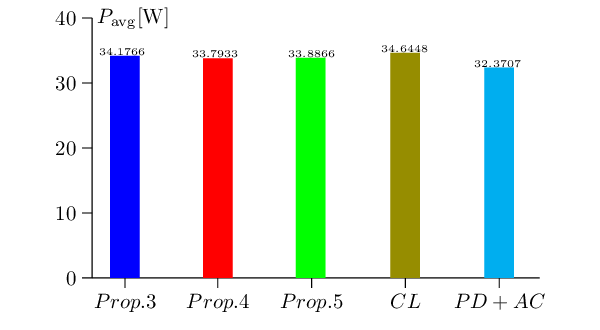}}}\\
\subfloat[Root mean square value of parametric errors $\tilde{\boldsymbol\theta}$ defined by $\tilde{\boldsymbol\theta}(t)_{{\rm RMS}}$ in \eqref{eq:normth}.]{%
\resizebox*{7.5cm}{!}{\includegraphics{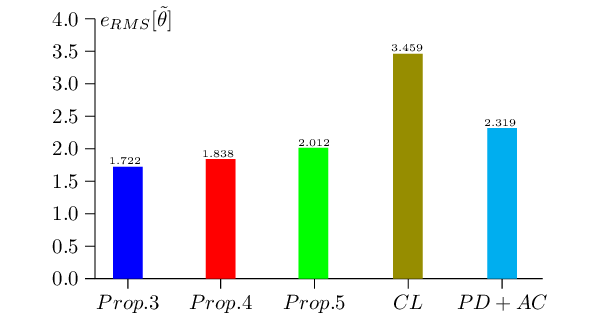}}}
\caption{Comparative performance index}
\label{Fig:9}
\end{figure}
\par Second, it is proposed to compute the root-mean square value of the parameter estimation error given by
\begin{equation}
\label{eq:normth}
{\tilde{\boldsymbol\theta}(t)}_{{\mathrm{RMS}}}= \sqrt{\frac{1}{T}\int_{0}^{T}\tilde{\boldsymbol\theta}(t)^T\tilde{\boldsymbol\theta}(t)dt},
\end{equation}
where $T = 20$ [s] represents the total experimentation time. The smaller $\tilde{\boldsymbol\theta}(t)_{{\rm RMS}}$ means again a better performance, which was obtained for the controllers of Propositions \ref{pro3}-\ref{pro5} over the CLAC and PD+AC (see Figure \ref{Fig:9}c).
\par Third, it is proposed to compute the average supplied power given by the formula:
\begin{equation}
\label{eq:PE}
P_{\rm{avg}}[\dot{\bm{q}}(t),\boldsymbol\tau(t)] = \frac{1}{T}\int_{0}^Tp(t)dt \hspace{0.25cm}\mathrm{[W]},
\end{equation}
where $T = 20 \ \mathrm{[s]}$ represents the total experimentation time and $p(t)$ is the instantaneous supplied power given by $p(t)=\dot{\bm{q}}(t)^T\boldsymbol\tau(t)$. Figure \ref{Fig:8} shows the plot of the instantaneous power supplied in all experiments, while Figure \ref{Fig:9}b shows the resultant index $P_{\rm{avg}}$ obtained from their respective plots. The smaller $P_{\rm{avg}}$ means the  better performance, the reason is that this quantity represents the real power supplied from the motor-drivers throughout the actuators in the system, that is, the power necessary to accomplish a task. When comparing the average supplied power $P_{\rm{avg}}$ of multiple controllers, providing the same control task (e.g verifying the same control objective), the controller with an smaller value of $P_{\rm{avg}}$ verifies the control task with less effort, thus considered had better performance achieving the control task. The better performance was obtained from the proposed controller PD+AC. However note that all proposed controllers of Proposition \ref{pro3}- \ref{pro5} performed better against the CL, as Figure \ref{Fig:9}b shows.
\par These indices show that, for two of them, namely ${\tilde{\boldsymbol\theta}(t)}_{\mathrm{RMS}}$ in \eqref{eq:normth} and $P_{\rm{avg}}$ in \eqref{eq:PE}, the controllers of Propositions \ref{pro3}--\ref{pro5} achieved better overall performance than the CLAC and PD+AC controllers. This analysis, together with the reductions in computational cost and memory usage extensively discussed in subsection \ref{sub:computational}, demonstrates the superior performance of the three proposed controllers.

\section{Conclusions}\label{secc:Conclusions}
Three adaptive global trajectory-tracking controllers for fully actuated Euler--Lagrange systems have been introduced in this paper. In all three cases, the joint position tracking objective is achieved. For the first and third proposals, it has been shown that the joint position and velocity tracking errors, as well as the parameter estimation error, converge to zero exponentially under a weak interval excitation condition, using the power-balance parameterization and the LS$+$DREM estimator. For the second proposal, it has been shown that the joint position and velocity tracking errors and the parameter estimation error converge asymptotically to zero under an interval excitation condition when using the energy-based regressor. Experimental validation on a two-degree-of-freedom robotic arm was conducted to evaluate the performance of the three controllers, and the results were compared with those of a composite learning-based adaptive controller (CLAC) that guarantees similar stability properties under comparable excitation conditions. The experimental tests showed that the proposed estimation algorithm is able to estimate the unknown constant parameters of the mechanical system and achieves the trajectory-tracking control objective, as predicted by the stability analysis. Finally, the performance-index results indicate that, compared to the CLAC and the classical PD$+$AC controller, the proposed controllers achieve improved performance in selected indices and a significantly lower computational cost.

\appendix

\section{Matrosov’s Theorem}
Consider the differential equation
\begin{equation}\label{eqn:xdot}
\dot{\bm{x}}=\bm{f}(\bm{x},t),
\end{equation}
where $\bm{x}\in\mathbb{R}^n,t\in\mathbb{R}$ is the time and $\bm{f}$ is a continuous function $\bm{f}:I\times \Xi\to \mathbb{R}^n$, where $I=[t_0,\infty)$ for some $t_0\in\mathbb{R}$ and $\Xi$ is an open connected set in $\mathbb{R}^n$, containing the origin. Assume that $\bm{f}(t,\bm{0})=0 \ \forall t\in I$, so that the origin is an equilibrium point for the differential equation \eqref{eqn:xdot}. Matrosov's theorem then states:
\begin{thm}{\cite{rouche}}
Let there exist two $C^{1}$ functions $V : I \times \Xi \to \mathbb{R}$,
$W : I \times \Xi \to \mathbb{R}$,
a $C^{0}$ function $V^{*} : \Xi \to \mathbb{R}$,
three functions $a, b, c$ of class $\mathcal{K}$ and two constants $S > 0$ and $T > 0$
such that, for every $(t,\bm{x}) \in I \times \Xi$:

\begin{itemize}
    \item[(i)] $a(\|\bm{x}\|) \le V(t,\bm{x}) \le b(\|\bm{x}\|)$;
    \item[(ii)] $\dot{V}(t,\bm{x}) \le V^{*}(\bm{x}) \le 0$;
        \quad $E := \{ \bm{x} \in \Xi : V^{*}(\bm{x}) = 0 \}$;
    \item[(iii)] $|W(t,\bm{x})| < S$;
    \item[(iv)] $\max\!\left( d(\bm{x},E),\, |\dot{W}(t,\bm{x})| \right) \ge c(\|\bm{x}\|) > 0$;
    \item[(v)] $\| f(t,\bm{x}) \| < T$.
\end{itemize}

Choosing $\alpha > 0$ such that $\bar{B}_{\alpha} \subset \Xi$,
let us put for every $t \in I$
\begin{equation}
V^{-1}_{t,\alpha} := \{ \bm{x} \in \Xi : V(t,\bm{x}) \le a(\alpha) \}.
\end{equation}

\noindent Then:
\begin{itemize}
    \item[(i)] For any $t_{0} \in I$ and any $\bm{x}_{0} \in V^{-1}_{t_{0},\alpha}$,
    any solution $\bm{x}(t)$ of (72), passing through $(\bm{x}_{0}, t_{0}) \in I \times \Xi$,
    tends to zero uniformly in $t_{0}$ and $\bm{x}_{0}$ as $t \to \infty$.
    \item[(ii)] The origin is uniformly asymptotically stable.
\end{itemize}
\end{thm}

The theoretical foundations of this result were originally established in \cite{matrosov}.
A notable feature of this theorem is that it provides conditions under which a Lyapunov function with a negative semi-definite derivative is still sufficient to guarantee asymptotic stability. The key idea lies in introducing a bounded auxiliary function that prevents the system trajectories from remaining in the set where the derivative of the Lyapunov function vanishes. To simplify the verification of condition (iv), we make use of the following lemma from \cite{paden}:

\begin{lemma}{\cite{paden}}
Condition (iv) of Matrosov’s theorem is satisfied if the following conditions hold:
\begin{itemize}
    \item[(iv.a)] $\dot{W}(\bm{x},t)$ is continuous in both arguments and depends on time as
    $$
    \dot{W}(\bm{x},t) = g(x, \beta(t)),
    $$
    where $g$ is continuous in both variables and $\beta(t)$ is continuous with image in a bounded set $K_{1}$. (For simplicity, we simply say that $\dot{W}(\bm{x},t)$ depends on time continuously through a bounded function.)

    \item[(iv.b)] There exists a class $\mathcal{K}$ function $k$ such that
    $$
    |\dot{W}(\bm{x},t)| \ge k(\|\bm{x}\|), \qquad \forall \bm{x} \in E,\; t \ge t_{0}.
    $$
\end{itemize}
\end{lemma}

For a proof, see \cite{paden}.

\section{Proof of Proposition \ref{pro4}}

To establish the result, it suffices to verify the hypotheses of Matrosov’s Theorem. Although most of the required calculations are conceptually straightforward, they become rather cumbersome in practice, particularly when checking condition (iv.a) of Lemma 3. Overall, the proof relies extensively on the boundedness theorem, Lemma 3, and several standard properties of matrix norms.

\par First, let us define the state vector as
\begin{equation*}
\bm{x}(t)=\begin{bmatrix}\tilde{\bm{q}}(t)^T & \dot{\tilde{\bm{q}}}(t)^T & \tilde{\boldsymbol\theta}(t)^T\end{bmatrix}^T\in\mathbb{R}^{2{n_q}+w},
\end{equation*}
whereas the vector field $f(\bm{x},t)$ is implicitly defined by \eqref{dyn1}.
\begin{itemize}
\item[i)] By Assumptions \ref{assum:M}–\ref{assum:u}, condition (i) is satisfied with $V(\bm{x},t)=\mathbb{W}$. Moreover, since $\mathbb{W}$ is radially unbounded, we have $\Xi=\mathbb{R}^{2{n_q}+w}$.

\item[ii)] From \eqref{dotU1}, the function
\begin{equation*}
V^\ast(\bm{x}):= - \lambda_{\min} \{ K_D \} \| \dot{\tilde{\bm{q}}} \|^2 -k_I \Delta_{\min} \|\tilde{\boldsymbol\theta}\|^2.
\end{equation*}
The set $E$ is therefore given by
\begin{align*}
E\triangleq& \Bigg\{ \begin{bmatrix}\tilde{\bm{q}}\\ \dot{\tilde{\bm{q}}} \\ \tilde{\boldsymbol\theta}\end{bmatrix}\in\Xi : V^\ast(\bm{x})=0  \Bigg\},\\
=&\Bigg\{ \begin{bmatrix}\tilde{\bm{q}}\\ \dot{\tilde{\bm{q}}} \\ \tilde{\boldsymbol\theta}\end{bmatrix}\in\Xi : \dot{\tilde{\bm{q}}}=\bm{0}_n,\; \tilde{\boldsymbol\theta}=\bm{0}_w,\; \tilde{\bm{q}}\in\mathbb{R}^n \Bigg\}.
\end{align*}

\item[iii)] Define the function $W=\ddot{\mathbb{W}}$. Taking the time derivative of \eqref{dotU11} along the trajectories of the closed-loop system \eqref{dyn1} yields:
\begin{align*}
W=&-2\dot{\tilde{\bm{q}}}^T K_D \ddot{\tilde{\bm{q}}}-2k_I\Delta(t)\dot{\Delta}(t)\tilde{\boldsymbol\theta}^T\tilde{\boldsymbol\theta}\\
&-2k_I\Delta^2(t)\tilde{\boldsymbol\theta}^T\dot{\tilde{\boldsymbol\theta}}\\
=&2\dot {\tilde{\bm{q}}}^\top K_D M^{-1}(\bm{q}) \Bigg( [C(\bm{q}, \dot{\bm{q}}) +K_D] \dot {\tilde{\bm{q}}}\\
&+K_P \tilde{\bm{q}} -k_I Y_n \tilde{\boldsymbol\theta} \Bigg)+2k_I \Delta^2(t) \tilde {\boldsymbol\theta}^\top \Big(k_I \Gamma Y_n^\top \dot{\tilde{\bm{q}}}\\
 &+k_I\Gamma \Delta^2(t)\tilde{\boldsymbol\theta}\Big) -2k_I \Delta(t) \dot \Delta(t) \tilde{\boldsymbol\theta}^T\tilde{\boldsymbol\theta}.
\end{align*}
At this point, conditions (i)–(ii) of Matrosov's theorem (already established) imply that $\tilde{\bm{q}},\, \dot{\tilde{\bm{q}}}\in\mathcal{L}_\infty^{n_q}$ and $\tilde{\boldsymbol\theta}\in\mathcal{L}_\infty^{w}$. Since, by assumption, $\bm{q}_{\star}(t)$ and its derivatives are bounded, it follows that $\bm{q}(t),\, \dot{\bm{q}}(t)\in\mathcal{L}_\infty^{n_q}$.
Moreover, Assumptions \ref{assum:M}–\ref{assum:u} ensure that the matrices $M(\bm{q})$ and $C(\bm{q},\dot{\bm{q}})$ are bounded. Consequently, the regressor $Y_n$ is also bounded, since it is composed of the bounded matrices $M(\bm{q})$, $C(\bm{q},\dot{\bm{q}})$, and the bounded vectors $\bm{g}(\bm{q})$ and $\nabla F_R$. From these boundedness relations and \eqref{dyn1}, it follows that $\ddot{\tilde{\bm{q}}}\in\mathcal{L}_\infty^{n_q}$. Furthermore, as shown in \cite{ORTROMARA}, the signal $\Delta(t)\in\mathcal{L}_\infty$ because all estimator signals \eqref{LSD1}–\eqref{LSD3} are bounded. From \eqref{dyn2}, this also implies that $\dot{\tilde{\boldsymbol\theta}}\in\mathcal{L}_\infty^{w}$. Moreover, since $\tilde{\boldsymbol\theta}$ is bounded and $\boldsymbol\theta$ is constant, it follows that $\hat{\boldsymbol\theta}\in\mathcal{L}_\infty^{w}$. This, in turn, implies that the control input $\boldsymbol\tau$ of \eqref{newt1} is bounded, and from \eqref{robdyn} we conclude that $\ddot{\bm{q}}\in\mathcal{L}_\infty^{n_q}$. Using Jacobi's formula \cite{magnus1999matrix} and the definition \eqref{eqn:Delta}, the time derivative of $\Delta$ is computed as
\begin{equation}\label{eqn:deltadot}
\frac{d}{dt}\Delta(t)=\Delta(t)\mathrm{tr}\{A^{-1}\dot{A}\}=\mathrm{tr}\Big\{\mathrm{Adj}\{A\}\dot{A}\Big\},
\end{equation}
where $A:=I_{w\times w}-z(t)f_0F(t)$. Since $\dot{A}=-\dot{z}f_0F(t)-zf_0\dot{F}(t)$ is bounded and $A$ is bounded, we conclude that $A^{-1}$ and $\dot{\Delta}(t)$ are also bounded.  Since all terms on the right-hand side of the expression for $W(\bm{x},t)$ are bounded, we conclude that $|W(\bm{x},t)|<S$.

\item[iv)] By Lemma 3, it suffices to verify conditions (iv.a) and (iv.b). Computing $\dot{W}$ yields:
\begin{align}
\nonumber
\dot{W}=& -2\ddot{\tilde{\bm{q}}}^T K_D\ddot{\tilde{\bm{q}}}
-2\dot{\tilde{\bm{q}}}^T K_D\dddot{\tilde{\bm{q}}}
-2k_I\dot{\Delta}^2(t)\tilde{\boldsymbol\theta}^T\tilde{\boldsymbol\theta}\\
\nonumber
&-2k_I\Delta(t)\ddot{\Delta}(t)\tilde{\boldsymbol\theta}^T\tilde{\boldsymbol\theta}
-4k_I\Delta(t)\dot{\Delta}(t)\tilde{\boldsymbol\theta}^T\dot{\tilde{\boldsymbol\theta}}\\
\nonumber
&-4k_I\Delta(t)\dot{\Delta}(t)\tilde{\boldsymbol\theta}^T\tilde{\boldsymbol\theta}-2k_I\Delta^2(t)\dot{\tilde{\boldsymbol\theta}}^T\dot{\tilde{\boldsymbol\theta}}\\
\label{eqn:dotW1}
&-2k_I\Delta^2(t)\tilde{\boldsymbol\theta}^T\ddot{\tilde{\boldsymbol\theta}}.
\end{align}
Except for the terms $\dddot{\tilde{\bm{q}}}$, $\ddot{\tilde{\boldsymbol\theta}}$, and $\ddot{\Delta}(t)$, all terms on the right-hand side of \eqref{eqn:dotW1} are known to be continuous in the tracking error and depend continuously on time through bounded functions. To establish (iv.a), we first show that $\dddot{\tilde{\bm{q}}}$ is continuous with respect to the tracking error and depends continuously on time through bounded functions. Differentiating \eqref{dyn1} yields:
\begin{align}
\nonumber
\frac{d}{dt}\ddot{\tilde{\bm{q}}}=& \dot{M}^{-1}\Bigg[-(C+K_D)\dot{\tilde{\bm{q}}}-K_p\tilde{\bm{q}}+k_IY_n\tilde{\boldsymbol\theta} \Bigg]\\
\nonumber
&+M^{-1}\Bigg[-\dot{C}\dot{\tilde{\bm{q}}}-C\ddot{\tilde{\bm{q}}}-K_D\ddot{\tilde{\bm{q}}}-K_p\dot{\tilde{\bm{q}}}\\
\label{aux2}
&+k_I\dot{Y}_n\tilde{\boldsymbol\theta}+k_IY_n\dot{\tilde{\boldsymbol\theta}}\Bigg]
\end{align}
By Assumptions \ref{assum:M}–\ref{assum:u}, and the fact that $\dddot{q}_\star$ is continuous and bounded, the terms $\dot{M}$, $\frac{d}{dt}C(q,\dot q)$, and $\dot{Y}_n(q,\dot q, \dot q_\star,\ddot q_\star)$ are all continuous. Hence, $\dddot{\tilde{\bm{q}}}$ is continuous with respect to the tracking error and depends continuously on time through bounded functions, because all terms on the right-hand side of \eqref{aux2} are continuous and depend on time only through bounded quantities. For the term $\ddot{\Delta}(t)$, differentiating \eqref{eqn:deltadot} gives:
\begin{equation}\label{eqn:ddotDelta}
\ddot{\Delta}(t)=\dot{\Delta} \mathrm{tr}\{A^{-1}\dot{A}\}+\Delta\, \mathrm{tr}\{ \dot{A}^{-1}+A^{-1}\ddot{A}\},
\end{equation}
with
\begin{equation}
\ddot{A}=-\ddot{z}(t)f_0F(t)-2\dot{z}f_0\dot{F}(t)-z(t)f_0\ddot{F}(t).
\end{equation}
Furthermore, from \eqref{LSD1}–\eqref{LSD3}, the time derivative of $\dot{F}$ yields
\begin{align}
\nonumber
\ddot{F}=&-\alpha\dot{F}\Omega\Omega^TF-\alpha F\dot{\Omega}\Omega^TF-\alpha F \Omega\dot{\Omega}^TF\\
&-\alpha F \Omega \Omega^T \dot{F}-\frac{\beta_0}{\rho}\frac{d}{dt}[\| F(t) \|]+\beta\dot{F}.
\end{align}
Since the spectral norm of a matrix is given by $\|F(t)\|=\sigma_{\max}(F(t))$, where $\sigma_{\max}(F)$ denotes the largest singular value of $F$, its time derivative satisfies
\[
\frac{d}{dt}\sigma_{\max}(F)=u^{\top}\dot F\,v,
\]
where $u$ and $v$ are the left and right unit singular vectors associated with $\sigma_{\max}(F)$, respectively. Hence,
\[
\left|\frac{d}{dt}\|F(t)\|\right|
=\left|u^{\top}\dot F\,v\right|
\le \|\dot F\|.
\]
Since $\dot{F}$ is bounded (because all estimator signals \eqref{LSD1}–\eqref{LSD3} are bounded, as proven in \cite{ORTROMARA}), it follows that $\|\dot{F}\|$ is also bounded, and thus $\frac{d}{dt}\|F\|$ is bounded as well. Moreover, from \eqref{newlreel} we observe that $\dot{\Omega}$ is continuous and bounded, since the right hand terms of \eqref{eqn:12} are continuous and bounded. Finally, the term $\ddot{z}$ is computed from \eqref{LSD3} as
\begin{equation}
\ddot{z}=\frac{\beta}{\rho}\frac{d}{dt}\big[\|F(t)\|\big]z-\beta\dot{z},
\end{equation}
which is bounded using the above arguments. The preceding results show that the term $\ddot{\Delta}$ in \eqref{eqn:dotW1}, given explicitly in \eqref{eqn:ddotDelta}, is bounded. Finally, the time derivative of $\dot{\tilde{\boldsymbol\theta}}$ in \eqref{dyn2} is
\begin{equation}
\frac{d}{dt}\dot{\tilde{\boldsymbol\theta}}
=-k_I\Gamma\Big[\dot{Y}_n^{\,T}\dot{\tilde{\bm{q}}}
+Y_n^{\,T}\ddot{\tilde{\bm{q}}}
+2\Delta(t)\dot{\Delta}(t)\tilde{\boldsymbol\theta}
+\Delta^2(t)\dot{\tilde{\boldsymbol\theta}}\Big],
\end{equation}
and all right-hand-side terms have already been shown continuous in the previous arguments. Therefore, $\frac{d}{dt}\dot{\tilde{\boldsymbol\theta}}$  is continuous with respect to the tracking error and depends continuously on time through bounded functions. Consequently, all terms in $\dot{W}$ are continuous in the tracking error and depend on time only through bounded functions, completing the verification of condition (iv.a).

\par To prove condition (iv.b), note that inside the set $E$ the expression for $\dot{W}$ reduces to:
\begin{align*}
\dot{W}=&-2\ddot{\tilde{\bm{q}}}^T K_D\ddot{\tilde{\bm{q}}}
        -2k_I\Delta^2(t)\dot{\tilde{\boldsymbol\theta}}^T\dot{\tilde{\boldsymbol\theta}}\\
      =&-2[-M^{-1}K_p\tilde{\bm{q}}]^T K_D [-M^{-1}K_p\tilde{\bm{q}}]\\
      =&-2\tilde{\bm{q}}^T K_p^T M^{-1}K_D M^{-1}K_p\tilde{\bm{q}}\\
       &< 0,\ \forall x \in E \subset \Xi,
\end{align*}
which is clearly a {\em{non-zero definite}} function on $E$ \cite{paden}. Moreover,
\[
|\dot{W}|\geq \lambda_{\min}\{K_p^T M^{-1} K_D M^{-1} K_p\}\,\|\tilde{\bm{q}}\|^2
\geq k\big(\|\bm{x}(t)\|\big),
\]
which establishes condition (iv.b).

\item[v)] Finally, since we have already proven that $\ddot{\tilde{\bm{q}}}$ is bounded, it follows that $\|f(\bm{x},t)\|$ is bounded as well.
\end{itemize}

Thus, all conditions of Matrosov’s theorem are satisfied. Therefore, the origin of the closed-loop system is a globally uniformly asymptotically stable equilibrium point. This completes the proof of Proposition~4.

\section*{Acknowledgments}
This work was partially supported by SECIHTI under Grants CVU 1106239, by TecNM projects and by Red Internacional de Control y Cómputo Aplicados del TecNM (RICCA / TecNM).

{}


\begin{thebibliography}{}

\bibitem{ZHOUeta}
Q. Zhou, S. Zhao, H. Li, R. Lu, and C. Wu,
``Adaptive Neural Network Tracking Control for Robotic Manipulators With Dead Zone,''
\emph{IEEE Transactions on Neural Networks and Learning Systems},
vol. 30, pp. 3611--3620, 2019.

\bibitem{DUCetal}
D. M. Le, O. S. Patil, P. M. Amy, and W. E. Dixon,
``Integral Concurrent Learning-Based Accelerated Gradient Adaptive Control of Uncertain Euler-Lagrange Systems,''
in \emph{IEEE American Control Conference},
pp. 806--811, 2022.
%
\bibitem{JIN}
X. Jin,
``Iterative learning control for non-repetitive trajectory tracking of robot manipulators with joint position constraints and actuator faults,''
\emph{International Journal of Advanced Robotic Systems},
vol. 31, pp. 859--875, 2016.
%
\bibitem{GUOPAN}
K. Guo and Y. Pan,
``Composite adaptation and learning for robot control: A survey,''
\emph{Annual Reviews in Control},
vol. 55, pp. 279--290, 2023.
%
\bibitem{CHELIUSLO}
C. C. Cheah, C. Liu, and J. J. Slotine,
``Adaptive Tracking Control for Robots with Unknown Kinematic and Dynamic Properties,''
\emph{International Journal of Robotics Research},
vol. 25, pp. 283--296, 2006.
%
\bibitem{YANGetal}
C. Yang, T. Teng, B. Xu, Z. Li, J. Na, and C. Su,
``Global adaptive tracking control of robot manipulators using neural networks with finite-time learning convergence,''
\emph{International Journal of Control, Automation and Systems},
vol. 15, pp. 1916--1924, 2017.

\bibitem{BRAGetal}
D. Braganza, W. E. Dixon, D. M. Dawson, and B. Xian,
``Tracking Control for Robot Manipulators with Kinematic and Dynamic Uncertainty,''
in \emph{44th IEEE Conference on Decision and Control},
pp. 5293--5297, 2005.
%
\bibitem{DIXetal}
W. E. Dixon, M. S. de Queiroz, F. Zhang, and D. M. Dawson,
``Tracking Control of Robot Manipulators with Bounded Torque Inputs,''
\emph{Robotica},
vol. 17, pp. 121--129, 2009.

\bibitem{BERORTNIJ}
H. Berghuis, R. Ortega, and H. Nijmeijer,
``A robust adaptive robot controller,''
\textit{IEEE Transactions on Robotics and Automation},
vol.~9, no.~6, pp.~825--830, 1993.

\bibitem{DUANAR}
M. Duarte and K. Narendra,
``Combined direct and indirect approach to adaptive control,''
\textit{IEEE Transactions on Automatic Control},
vol.~34, no.~10, pp.~1071--1075, 1989.

\bibitem{GOOSINbook}
G. Goodwin and K. Sin,
\textit{Adaptive Filtering Prediction and Control},
Prentice-Hall, 1984.

\bibitem{KHOKAN}
P. K. Khosla and T. Kanade,
``Experimental evaluation of nonlinear feedback and feedforward control schemes for manipulators,''
\textit{International Journal of Robotics Research},
vol.~7, pp.~18--28, 1988.

\bibitem{KRERIE}
G. Kreisselmeier and G. Rietze-Augst,
``Richness and excitation on an interval---with application to continuous-time adaptive control,''
\textit{IEEE Transactions on Automatic Control},
vol.~35, no.~2, pp.~165--171, 1990.

\bibitem{ORTetalbook}
R. Ortega, A. Loria, P. J. Nicklasson, and H. Sira-Ramirez,
\textit{Passivity-Based Control of Euler--Lagrange Systems},
Springer-Verlag, Berlin, Communications and Control Engineering, 1998.

\bibitem{ORTROMARA}
R. Ortega, J. G. Romero, and S. Aranovskiy,
``A new least squares parameter estimator for nonlinear regression equations with relaxed excitation conditions and forgetting factor,''
\textit{Systems \& Control Letters},
vol.~169, 2022.


\bibitem{PANORTMOY}
E. Panteley, R. Ortega, and P. Moya,
``Overcoming the detectability obstacle in certainty equivalence adaptive control,''
\textit{Automatica},
vol.~38, no.~7, pp.~1125--1132, 2002.

\bibitem{ROMORTBOB}
J. G. Romero, R. Ortega, and A. Bobtsov,
``Parameter estimation and adaptive control of Euler-Lagrange systems using the power balance equation parameterisation,''
\textit{International Journal of Control},
vol.~96, no.~2, pp.~475--487, 2023.

\bibitem{SASBODbook}
S. Sastry and M. Bodson,
\textit{Adaptive Control: Stability, Convergence and Robustness},
Prentice-Hall, New Jersey, 1989.

\bibitem{SLOLItac}
J. J. E. Slotine and W. Li,
``Adaptive manipulator control: a case study,''
\textit{IEEE Transactions on Automatic Control},
vol.~33, no.~11, pp.~995--1003, 1988.

\bibitem{SLOLIaut}
J. J. E. Slotine and W. Li,
``Composite adaptive control of robot manipulators,''
\textit{Automatica},
vol.~25, no.~4, pp.~509--519, 1989.

\bibitem{SPOHUTVIDbook}
M. W. Spong, S. Hutchinson, and M. Vidyasagar,
\textit{Robot Modeling and Control},
Wiley, 2020.

\bibitem{WANORTBOB}
L. Wang, R. Ortega, and A. Bobtsov,
``Observability is sufficient for the design of globally exponentially convergent state observers for state-affine nonlinear systems,''
\textit{Automatica},
vol.~25, no.~110838, 2023.

\bibitem{Cicese41}
F. Reyes and R. Kelly,
``Experimental evaluation of identificaction schemes on a direct--drive robot,''
\textit{Robotica},
vol.~15, no.~5, pp.~563--571, 1997.

\bibitem{Cicese42}
F. Reyes and R. Kelly,
``Experimental evaluation of model-based controllers on a drive robot arm,''
\textit{Mechatronics},
vol.~11, no.~3, pp.~267--182, 2001.

\bibitem{ControlOfRobots}
R. Kelly, V. Santibáñez, and A. Loria,
\textit{Control of Robot Manipulators in Joint Space},
Springer London, 2006.

\bibitem{WinMechLab}
R. Campa, R. Kelly, and V. Santibáñez,
``Windows-based real-time control of direct-drive mechanisms: platform description and experiments,''
\textit{Mechatronics},
vol.~14, no.~9, pp.~1021--1036, Nov. 2004.
%





\bibitem{yongpinpan}
Y. Pan and H. Yu,
``Composite learning robot control with guaranteed parameter convergence,''
\textit{Automatica},
vol.~89, pp.~398--406, 2018.







\bibitem{guo2018composite2}
K. Guo, Y. Pan, and H. Yu,
``Composite learning robot control with friction compensation: A neural network-based approach,''
\textit{IEEE Transactions on Industrial Electronics},
vol.~66, no.~10, pp.~7841--7851, 2018.


\bibitem{magnus1999matrix}
J. R. Magnus and H. Neudecker,
\textit{Matrix Differential Calculus with Applications in Statistics and Econometrics},
Wiley, 1999.

\bibitem{khalil2002nonlinear}
H. K. Khalil,
\textit{Nonlinear Systems},
Prentice Hall, 2002.

\bibitem{BERGHUIS}
H. Berghuis, H. Roebbers, and H. Nijmeijer,
``Experimental comparison of parameter estimation methods in adaptive robot control,''
\textit{Automatica},
vol.~31, no.~9, pp.~1275--1285, 1995.


\bibitem{GUO2020}
K. Guo, Y. Pan, D. Zheng, and H. Yu,
``Composite learning control of robotic systems: A least squares modulated approach,''
\textit{Automatica},
vol.~111, p.~108612, 2020.

\bibitem{matrosov}
V. M. Matrosov,
``On the stability of motion,''
\textit{Journal of Applied Mathematics and Mechanics},
vol.~26, no.~5, pp.~1337--1353, 1962.

\bibitem{rouche}
N. Rouche, P. Habets, and M. Laloy,
\textit{Stability Theory by Liapunov's Direct Method},
vol.~4,
Springer, 1977.

\bibitem{paden}
B. Paden and R. Panja,
``Globally asymptotically stable PD$+$ controller for robot manipulators,''
\textit{International Journal of Control},
vol.~47, no.~6, pp.~1697--1712, 1988.

\bibitem{Cervantes2026}
L. Cervantes-Pérez, V. Santibáñez, J. G. Romero, R. Ortega, and J. Sandoval,
``Experimental 2 d.o.f. robot manipulator physical parameters estimation and indirect adaptive control using the power balance equation parametrisation,''
\textit{International Journal of Control},
pp.~1--15, 2026,
doi: 10.1080/00207179.2026.2638931.

\bibitem{ROMORT}
J. G. Romero and R. Ortega,
``Two high performance global tracking composite adaptive controllers for fully actuated Euler-Lagrange systems,''
\textit{IFAC-PapersOnLine},
vol.~58, no.~6, pp.~202--207, 2024.

\bibitem{raol}
J. R. Raol, G. Girija, and J. Singh,
\textit{Modelling and Parameter Estimation of Dynamic Systems},
vol.~65,
IET, 2004.

\bibitem{yoshikawa1985}
T. Yoshikawa,
``Manipulability of robotic mechanisms,''
\textit{The International Journal of Robotics Research},
vol.~4, no.~2, pp.~3--9, 1985.

\bibitem{yoshikawa1990}
T. Yoshikawa,
``Translational and rotational manipulability of robotic manipulators,''
in \textit{1990 American Control Conference},
pp.~228--233, 1990.

\bibitem{pamanes2018}
J. Pamanes and H. Moreno,
``Homogenization of the Jacobian matrix of manipulators by using inertial parameters,''
in \textit{IFToMM Symposium on Mechanism Design for Robotics},
pp.~219--226, 2018.

\bibitem{yu}
W. Yu, X. Li, and R. Carmona,
``A novel PID tuning method for robot control,''
\textit{Industrial Robot: An International Journal},
vol.~40, no.~6, pp.~574--582, 2013.

\bibitem{horn2012matrix}
R. A. Horn and C. R. Johnson,
\textit{Matrix Analysis},
Cambridge University Press, 2012.

\bibitem{golub2013matrix}
G. H. Golub and C. F. Van Loan,
\textit{Matrix Computations},
JHU Press, 2013.


\end{thebibliography}

\end{document}